\newif\iflong
\longtrue
\iflong
\documentclass{article}
\else
\documentclass[a4paper]{lipics-v2021}
\fi

\usepackage[utf8]{inputenc}

\usepackage[toc,page]{appendix}
\iflong
\usepackage[margin=1in]{geometry}
\usepackage[bookmarks, colorlinks=true, plainpages = false, citecolor = blue,linkcolor=red,urlcolor = blue, filecolor = blue,pagebackref]{hyperref}
\usepackage{url}
\else
\usepackage{hyperref}
\fi
\usepackage{amsmath,amsfonts,amsthm,amssymb,bm, verbatim,dsfont,mathtools}
\usepackage{color,graphicx,appendix}
\usepackage{amsfonts}
\usepackage{bbm}
\iflong
\usepackage{subfigure}
\usepackage{etoolbox}
\usepackage{tikz}
\usepackage{xr,xspace}
\usepackage{todonotes}
\usepackage{caption,soul}
\fi
\usepackage{paralist}
\usepackage[ruled,vlined,linesnumbered]{algorithm2e}
\usepackage{cite}

\newcommand*\squeezespaces[1]{% %% <- #1 is a number between 0 and 1
  \thickmuskip=\scalemuskip{\thickmuskip}{#1}%
  \medmuskip=\scalemuskip{\medmuskip}{#1}%
  \thinmuskip=\scalemuskip{\thinmuskip}{#1}%
  \nulldelimiterspace=#1\nulldelimiterspace
  \scriptspace=#1\scriptspace
}
\newcommand*\scalemuskip[2]{%
  \muexpr #1*\numexpr\dimexpr#2pt\relax\relax/65536\relax
} %% <- based on  https://tex.stackexchange.com/a/198966/156366

\newif\ifediting\editingtrue
\newif\ifcamera\camerafalse

\newcommand{\adverse}{b}
\newcommand{\mempool}{\mathsf{mempool}}
\newcommand{\negl}{\mathsf{negl}}

\newcommand{\ek}[1]{\mathrm{ek}_{#1}}
\newcommand{\vk}[1]{\mathrm{vk}_{#1}}
\newcommand{\alg}{\mathcal{A}}
\newcommand{\fl}{the full version}

\newcommand{\dif}{X}

\newcommand{\nce}{\mathsf{nce}}

\newcommand{\secp}{\lambda}
\newcommand{\orcl}{\mathcal{H}}

\newcommand{\func}{\mathcal{F}}
\newcommand{\eps}{\varepsilon}
\newcommand{\distas}[1]{\mathbin{\overset{#1}{\kern\z@\sim}}}%

\newcommand{\setup}{\mathtt{Setup}}
\newcommand{\eval}{\mathtt{Eval}}

\newcommand{\veri}{\mathtt{Verify}}

\newcommand{\sX}{\mathcal{X}}
\newcommand{\sY}{\mathcal{Y}}

\newcommand{\poly}{\text{poly}}
\newcommand{\perm}{\mathcal{P}}

\newcommand{\qq}[1]{{\nb{Quanquan: #1}}}

\ifcamera
\newcommand{\qq}[1]{}
\fi
\makeatletter
\iflong
\newtheorem{theorem}{Theorem}
\newtheorem{lemma}[theorem]{Lemma}
\newtheorem{proposition}[theorem]{Proposition}

\theoremstyle{definition}
\newtheorem{definition}[theorem]{Definition}

\newtheorem{remark}[theorem]{Remark}

\fi
\newtheorem{assumption}[theorem]{Assumption}

\newtheorem{property}[theorem]{Property}

\usepackage{xspace,prettyref}

\newcommand{\reals}{{\mathbb{R}}}

\newcommand{\naturals}{{\mathbb{N}}}

\definecolor{ForestGreen}{rgb}{0, 0.5, 0}

\newcommand{\blue}{\color{blue}}
\newcommand{\nb}[1]{{\sf\blue[#1]}}

\newcommand{\nbo}[1]{{\sf\color{orange}[#1]}}
\newcommand{\ls}[1]{\nbo{LS: #1}}
\newcommand{\Expect}{\mathbb{E}}
\newcommand{\expect}[1]{\mathbb{E}\left[ #1 \right]}

\newcommand{\Prob}{\mathbb{P}}

\newcommand{\prob}[1]{ \mathbb{P}\left\{ #1 \right\} }

\newcommand{\Binom}{{\rm Binom}}

\newcommand{\iid}{i.i.d.\xspace}
\newrefformat{eq}{(\ref{#1})}
\newrefformat{chap}{Chapter~\ref{#1}}
\newrefformat{sec}{Section~\ref{#1}}
\newrefformat{alg}{Algorithm~\ref{#1}}
\newrefformat{fig}{Fig.~\ref{#1}}
\newrefformat{tab}{Table~\ref{#1}}
\newrefformat{rmk}{Remark~\ref{#1}}
\newrefformat{clm}{Claim~\ref{#1}}
\newrefformat{def}{Definition~\ref{#1}}
\newrefformat{cor}{Corollary~\ref{#1}}
\newrefformat{lmm}{Lemma~\ref{#1}}
\newrefformat{prop}{Proposition~\ref{#1}}
\newrefformat{app}{Appendix~\ref{#1}}
\newrefformat{hyp}{Hypothesis~\ref{#1}}
\newrefformat{thm}{Theorem~\ref{#1}}
\newrefformat{ass}{Assumption~\ref{#1}}
\newrefformat{problem}{Problem~\ref{#1}}
\newcommand{\pth}[1]{\left( #1 \right)}

\newcommand{\sth}[1]{\left\{ #1 \right\}}

\newcommand{\abth}[1]{\left | #1 \right |}

\newcommand{\indc}[1]{{\mathbbm{1}{\left\{{#1}\right\}}}}

\newcommand{\calA}{{\mathcal{A}}}

\newcommand{\calC}{{\mathcal{C}}}

\newcommand{\calJ}{{\mathcal{J}}}

\newcommand{\calN}{{\mathcal{N}}}

\newcommand{\calW}{{\mathcal{W}}}

\newcommand{\calZ}{{\mathcal{Z}}}

\newcommand{\nn}[1]{\nb{NN: #1}}
\newcommand{\myparagraph}[1]{\noindent {\bf #1.}}
\usepackage[capitalise,noabbrev,nameinlink]{cleveref}

\title{The Power of Random Symmetry-Breaking \\in Nakamoto Consensus}
\author{
Lili Su \thanks{Correspondence author. Email: \url{l.su@northeastern.edu}.  }\\
ECE \\
Northeastern University  
\and 
Quanquan C. Liu\\
CSAIL\\
Massachusetts Institute of Technology%\\
\and 
Neha Narula \\
MIT Media Lab \\ 
Massachusetts Institute of Technology %\\
}

\iflong
\else
\author{Lili Su}{Northeastern University}{l.su@northeastern.edu}{}{} 
\author{Quanquan C. Liu}{Massachusetts Institute of Technology}{quanquan@mit.edu}{}{} \author{Neha Narula}{Massachusetts Institute of Technology}{narula@media.mit.edu}{}{}
\authorrunning{L. Su, Q. C. Liu, N. Narula} 
\Copyright{Lili Su, Quanquan C. Liu, Neha Narula}

\EventEditors{John Q. Open and Joan R. Access}
\EventNoEds{2}
\EventLongTitle{35th International Symposium on Distributed Computing (DISC 2021)}
\EventShortTitle{DISC 2021}
\EventAcronym{DISC}
\EventYear{2021}
\EventDate{Oct. 4--8, 2021}
\EventLocation{Freiburg, Germany}
\EventLogo{}
\SeriesVolume{42}
\ArticleNo{23}
\fi

\begin{document}
\sloppy
\maketitle

\begin{abstract}
Nakamoto consensus underlies the security of many of the world's largest cryptocurrencies, such as Bitcoin and Ethereum.
Common lore is that Nakamoto consensus only achieves consistency and liveness under a regime where the difficulty of its underlying mining puzzle is very high, negatively impacting overall throughput and latency. In this work, we study Nakamoto consensus under a wide range of puzzle difficulties, including very easy puzzles. 
We first analyze an adversary-free setting 
and show that, surprisingly, the common prefix of the blockchain grows quickly even with easy puzzles. 
In a setting with adversaries, we provide a small backwards-compatible change to Nakamoto consensus to achieve consistency and liveness with easy puzzles. 
Our insight relies on a careful choice of \emph{symmetry-breaking strategy}, which was significantly underestimated in prior work. 
We introduce a new method---\emph{coalescing random walks}---to analyzing the correctness of Nakamoto consensus under the uniformly-at-random symmetry-breaking strategy. 
This method is more powerful than existing analysis methods that focus on bounding the number of {\it convergence opportunities}. 
\end{abstract}

\section{Introduction}

Nakamoto consensus~\cite{bitcoin}, the elegant blockchain protocol that underpins many
cryptocurrencies, achieves consensus in a setting where nodes can join and leave the system without getting permission from a centralized authority.
Instead of depending on the identity of nodes, it achieves consensus by incorporating computational puzzles called 
proof-of-work~\cite{dwork} (also known as \emph{mining}) and using a simple longest-chain protocol.\footnote{We use "longest chain" to mean the one with the most proof-of-work given difficulty adjustments, not necessarily the one with the most blocks, though without considering difficulty adjustments they are the same.}
Nodes in a network maintain a
local copy of an append-only ledger and gossip messages to add to the ledger,
collecting many into a block. A block consists of the set of records
to add, a pointer to the previous block in the node's local copy of the ledger, and a
nonce, which is evidence the node has done proof-of-work, or solved a computational puzzle of sufficient difficulty, dependent on the block. 
The node then broadcasts its local chain to the network. Honest
nodes choose a chain they see with the most proof-of-work to
continue building upon.

Previous work defined correctness and liveness in proof-of-work
protocols (also referred to as the \emph{Bitcoin backbone}) using
three properties: \emph{common-prefix}, \emph{chain-quality}, and
\emph{chain-growth}~\cite{garay15,GKL17,PSS17}. Informally, common-prefix indicates that any two
honest nodes share a common prefix of blocks, chain-growth is the rate
at which the common prefix grows over time, and chain-quality represents the
fraction of blocks created by honest nodes in a chain. 
In previous work, achieving these properties critically relied on the
setting of the difficulty factor in the computational puzzles. We
express this as $p$, the probability that any node will solve the
puzzle in a given round. Previous work analyzing Nakamoto consensus
has shown that for consistency and liveness $p$ should be very small
in relation to the expected network delay and the number of nodes~\cite{garay15,PSS17}.
% $\Delta$, so that $p \leq \frac{1}{\rho n \Delta}, where $\rho n$ is the total amount of the adversary's computational power in the network. \qq{Would be good to settle on the definition of $\Delta$. It is still used as the network delay in terms of rounds in Section 3.} 
%
For example, mining difficulty in Bitcoin is set so that the network
is only expected to find a puzzle solution roughly once every ten
minutes.

%Intuitively, this is because in a network with high $p$ and large
%$\Delta$, many nodes will find puzzle solutions (and thus blocks) at
%the same time, leading to increased \emph{forking}: because of the delay in hearing about other nodes' chains, honest nodes will work on different chains at the same time, reducing chain growth. This was experimentally confirmed by Decker and Wattenhofer by observing the Bitcoin network~\cite{decker2013information}.
%and the adversary can delay them so honest nodes will
%``waste" computational power on blocks that end up discarded and do not
%make it into the longest chain. The adversary suffers no such delay
%and can dedicate all of its computational power to a strategy
%which, for example, forms a longer chain, outpacing the reduced
%effective computational power of the honest nodes.

Requiring a small $p$ increases block time, removing a parameter for improving transaction throughput. One way to compensate is by increasing block size, which could result in burstier network traffic and longer transaction confirmation times for users.
%
%leads to low transaction throughput and long transaction confirmation times---the original Bitcoin whitepaper suggests waiting six blocks, an hour, before considering a
%transaction finalized. Because of this, Bitcoin
%effectively processes 3-7 transactions per second, and Ethereum 10-15
%transactions per second, while other BFT consensus systems can process
%hundreds of thousands of operations per second~\cite{yin2018hotstuff}. Increasing $p$, if it could be done safely, would lead directly to higher throughput and lower latency in
%proof-of-work-based cryptocurrency systems.
%
Newer chains which do not use proof-of-work seem to favor short block times, probably because users value a fast first block confirmation: in EOS, blocks are proposed every 500 milliseconds~\cite{eos-block-time} and Algorand aims to achieve block finality in 2.5 seconds~\cite{algorand-finalize}, whereas in Bitcoin blocks only come out every ten minutes.

Common belief is that larger $p$ fundamentally 
constrains chain growth (i.e., the growth of the common prefix), even in the absence of an adversary,
due to the potential of increased \emph{forking}: nodes will find puzzle solutions (and thus blocks) at the same time; because of the delay in hearing about other nodes' chains nodes will build on different chains, delaying agreement.
Another common conjecture, explicitly mentioned in \cite{garay15}, is that the choice of 
\emph{symmetry-breaking strategies}, or ways honest nodes choose among multiple longest chains, is not relevant to correctness.

In this paper, we show that these common beliefs are incorrect. In particular, 
we show that when $p$ is beyond the well-studied region even the simple strategy of choosing among chains of equal length randomly
%uniformly-at-random symmetry-breaking strategy 
fosters chain growth, especially in the absence of adversaries. 
% \qq{We may want to mention that we require some assumptions
% on the adversary in our analysis or make the statement 
% above a bit weaker; something like "are incorrect
% under some settings" or perhaps something stronger than
% this.}
\smallskip

\noindent\textbf{Contributions.}
In this work, we formally analyze Nakamoto consensus under a wide range of $p$ including large $p$. 
%in the setting where $p$ is large. 
%
We confirm previous (informal) analysis that Nakamoto consensus requires
small $p$ in the presence of adversaries, but show that surprisingly,
it does not in a setting without adversaries, even if $p=1$ (all nodes
mine blocks every round) with a minor change in nodes' symmetry-breaking strategy. Previous work assumed the requirement of \emph{convergence opportunities}, a period when only one honest node mines a block, in order to achieve consistency~\cite{PSS17,KRS18}; we show that in fact convergence opportunities are \emph{not} required for common-prefix and chain growth. 
With an additional backwards-compatible modification to Nakamoto consensus, we can derive a bound on the chain growth for a wider range of $p$ (including large $p$)
%achieve higher chain growth with a large $p$ 
in a setting with adversaries. 
Our key idea in this modification is to
introduce a \emph{verifiable delay function}~\cite{BBBF18} to prevent the adversaries from extending a chain by multiple blocks in a round. 
%from mining multiple blocks in a round.
%
Our analysis is based on a new application of a well-known technique,
coalescing random walks. To our knowledge this is the first
application of coalescing random walks to analyze the common-prefix
and chain quality of Bitcoin and other proof-of-work protocols.
We thoroughly analyze Nakamoto consensus with the \emph{uniformly-at-random} symmetry-breaking strategy and discuss different symmetry-breaking strategies including \emph{first-seen}, \emph{lexicographically-first}, and \emph{global-random-coin}.

In summary, our contributions are as follows:

\begin{itemize}

\item A new approach for analyzing the confirmation time of the Bitcoin protocol under the uniformly-at-random
symmetry-breaking strategy in the adversarial-free setting via \emph{coalescing random walks}. Our analysis works for a new region of $p$, and shows that previous works' requirement for \emph{convergence opportunities} was unneeded.

\item New notions of \emph{adversarial advantages} and {\em coalescing opportunities} to provide a more general analysis of common-prefix and chain growth in Nakamoto consensus in the presence of adversaries. 
%{\red (for Lili to fix) Using this we can present a new tighter bound on chain growth.} \ls{Hard to say whether it is tight or not. What I am sure is the analysis is more general than the traditional ones. }

% \item A backwards-compatible modification to Nakamoto's protocol using verifiable delay functions
% that enables Nakamoto consensus to handle a larger $p$ than what can
% be handled in previous analyses. 
%We show that this addition enables faster confirmation times while maintaining common-prefix and chain quality.
%A thorough analysis of the first-seen and uniformly-at-random symmetry breaking strategies with large $p$ and in the presence of no adversaries.\nn{Is it worth emphasizing this contribution since it's in the presence of no adversaries? Don't we analyze it in the presence of adversaries too?}
\end{itemize}

\noindent\textbf{Related Work.}
Proofs-of-work were first put forth by Dwork and Naor~\cite{dwork}. Garay, Kiayias, and Leonardas~\cite{garay15} provided the first thorough analysis of Nakamoto's protocol in a synchronous static setting, introducing the ideas of \emph{common-prefix}, \emph{chain quality} and \emph{chain growth}. Later work~\cite{GKL17} extended the analysis to a variable difficulty function. Pass, Seeman, and shelat~\cite{PSS17} extended the idea of common-prefix to \emph{future self-consistency}, and provided an analysis of Nakamoto consensus in the semi-synchronous setting with an adaptive adversary. Several additional papers used this notion of future self-consistency~\cite{KRS18,ZTLWLX20}. \cite{PSS17,KRS18} relied on \emph{convergence opportunities}, or rounds where only one node mines a block, to analyze chain growth. In this work we show that convergence opportunities are \emph{not} required for chain growth, and relying on them underestimates chain growth with high $p$; in the adversary-free setting we show chain growth even with $p=1$ (no convergence opportunities; all nodes mine a block every round).
Other work considered the tradeoffs between chain growth and chain quality~\cite{PS17,PSS17,KP15,GKL17,ZP19}; however, to the best of our knowledge, none of these works considered different symmetry breaking strategies to enable faster chain growth while maintaining chain quality. In our paper, we thoroughly explore this domain.
Another line of work~\cite{eyal2014majority, sapirshtein2016optimal} considers how the uniformly-at-random symmetry breaking strategy affects incentive-compatible selfish mining attacks; our analysis applies to general attacks. 

Random walks have been used % in the past 
to analyze the probability of
consistency violations in proofs-of-stake protocols~\cite{BKMQR20}; 
%\nn{note that the prev cite does not mention coalescing random walks; they reference coupled bias random walks. Does that matter?} 
ours is the first work that uses coalescing random walks to analyze the common-prefix and chain quality of Bitcoin and other proof-of-work protocols.

\section{Model and Definitions}
\label{sec: model and def}
In this section, we present the specific model we use %in the paper 
and briefly describe the Bitcoin cryptosystem. 
We follow % Our model relies on 
the formalization %of models
presented in~\cite{GKL17,KRS18,PSS17}. 

\medskip

\myparagraph{Network and Computation Model}
%\qq{$\Delta$ is used differently than how it's used in the intro.}
Following previous work~\cite{garay15,GKL17,GKL20,PSS17,Ren19,ZTLWLX20},
we consider a synchronous network % model is the synchronous model
where nodes send messages in synchronous rounds, i.e., $\Delta=1$; equivalently, there is a global clock and the time is slotted into equal duration rounds.  
Each node has identical computing power. 
% Each node keeps a local version of the blockchain.  
Notably, the synchronous rounds assumption is significantly more relaxed than assuming $\Delta=0$.\footnote{In fact, the analysis based on Poisson race \cite{nakamoto2008bitcoin,bagaria2019prism} essentially assumes all mined blocks can be ordered in a globally consistent way, i.e., $\Delta=0$, which does not hold in our synchronous network model.} 
Our model operates
in the \emph{permissionless setting}. 
This means that any miner can join (or leave) the %corresponding blockchain 
protocol execution without getting permission from a centralized or distributed authority. 
 %(without getting permission from a centralized or distributed authority) as long as the number of participants that are executing the corresponding blockchain protocol remains at a certain level. 
%is $n$, %\nn{isn't it $\approx n$?}\qq{I think we're assuming it's exactly $n$...} 
% the total number of active participants is $n$, 
For ease of exposition, we assume the number of participants remains $n$.  
% A simple illustration is given in Figure \ref{fig:permissionless}. 
Our results can be easily generalized to handle perturbation in the population size 
% (i.e., the population size fluctuates around $n$) 
by a stochastic dominance argument 
as long as the population size does not deviate too far from $n$, and the proportion of Byzantine participants does not increase due to the perturbation. 
%\nn{I don't think this figure is needed and the space could be better used on technical content} \qq{We can remove if we run out of space.}
% All protocols discussed in this paper operate in the form of state machine replication (SMR) wherein  

\iffalse 
\begin{figure}
    \centering
    \includegraphics[width=0.5\textwidth]{PermissionlessSystem}
    \caption{This figure illustrates a permissionless system with $n=4$. At some time, the system contains 4 participants with one more miner (the yellow miner) that is about to join the system and one active miner (the blue miner) that wants to leave the system.}
    \label{fig:permissionless}
\end{figure}
\fi 

\medskip

\myparagraph{Adversary Model}\label{sec:adversary}
Throughout this paper, we assume that all Byzantine nodes are controlled by a \emph{probabilistic polynomial time (PPT) adversary} $\calA$ that can coordinate the behavior of all such nodes. $\calA$ operates in PPT which means they have access to random coins but can only use polynomial time to perform computations. At any time during the run of the protocol, $\calA$ can corrupt up to $\adverse$ nodes at any point in time where $\adverse$ 
is a parameter that 
is an input to the protocol.
%$\adverse$ is an algorithm input. % for some parameter $\adverse$ defined in the protocol. 
The corrupted nodes remain corrupted for the remainder of the protocol.
% Honest and Byzantine nodes may leave the network at any time
% as long as the total number of participants remains $n$.
%$\calA$ does \emph{not} have access to the internal state of honest nodes, such as secret cryptographic material of the honest nodes.
%Thus, $\calA$ cannot forge an honest node's signature (with all but negligible probability in the security parameter).  \nn{I don't think we need this; there are no secret keys}
Finally, $\calA$ cannot modify or delete the messages % broadcasts 
sent by honest nodes, but % $\calA$ 
can read all messages sent over the network and arbitrarily order the messages received by any honest nodes.

\subsection{Bitcoin Cryptosystem}
A \emph{blockchain protocol} is a stateful algorithm wherein each node  
maintains a local version of the blockchain $\calC$. 
Each honest node runs its own homogeneous version of the blockchain protocol. 
Nodes receive messages from the \emph{environment} $\calZ(1^\secp)$, where $\secp$ is the security parameter chosen based on the population size $n$. 
% of participants using the random oracle.
% The goal of an honest node is to 
% incorporate its message into its local chain and the local
% chains of other honest nodes. 
The environment is responsible for all the external factors related
to a protocol's execution. For example, it provides the value of $\adverse$ to the nodes. 
%and is responsible for nodes which enter or leave the protocol. 
Detailed description of the environment can be found in~\cite{PSS17}. 

The protocol begins by having the environment $\calZ$ initialize $n$
nodes. The protocol proceeds
in synchronous rounds; at each round $r$, each node 
receives a message from $\calZ$. 
% An honest node keeps a local chain and 
In each round, an honest node attempts
to mine a block containing its message to add to its local chain.
%$\calA$ can corrupt up to $b$ nodes at any point during the protocol; corrupted nodes remain corrupted for the remainder of the execution. \nn{already said this above}
We provide formal definitions of the Bitcoin
cryptosystem below.

\vskip 0.5\baselineskip
%\subsubsection{Blocks and Blockchains}
\noindent{\bf Blocks and Blockchains}\\
A blockchain $\calC~\triangleq~B_0B_1B_2\cdots B_{\ell}$ for some $\ell\in \naturals$ is a chain of blocks. Here $B_0$ is a predetermined \emph{genesis block} that all chains
must build from. 
%Generally speaking, we say a blockchain is valid if it satisfies certain properties, most importantly, that a certain amount of \emph{work} was done.
%Let $H(\cdot)$ and $G(\cdot)$ be cryptographic hash functions with output in $\{0, 1\}^{\kappa}$. 
% In the Bitcoin blockchain, 
A {\em block} $B_{\ell}$, for $\ell\ge 1$, is a triple $B_{\ell} = \langle s, x, \nce\rangle,$ 
% \begin{align*}
%  B = \langle s, x, \nce\rangle, 
% \end{align*}
where $s, x, \, \nce \, \in \sth{0,1}^{\ast}$ are three binary strings of arbitrary length. 
Specifically, $s$ is used to indicate this block's 
predecessor, $x$ is the text of the block containing the 
message (e.g.\ transactions) and other metadata, and $\nce$ is a \emph{nonce} chosen by a node. % to solve the proof-of-work puzzle. 
%-- is a binary string of arbitrary length.

\vskip 0.5\baselineskip 
\noindent{\bf Proofs-of-Work}\\
% \subsubsection{Proofs-of-Work}\label{sec:pow}
%A nonce is a randomly generated binary string $\nce \in \{0, 1\}^*$ of arbitrary length.  \nn{repeat}
The Bitcoin cryptosystem crucially uses nonces
as \emph{proofs-of-work} for determining whether a block can
be legally added to a chain.\footnote{Note that in practice, the nonce is effectively concatenated with a miner's public key (included in the \emph{coinbase} transaction) to ensure unique queries.  
The public key does \emph{not} need to be verified.
Importantly, this means that the miner can just generate a 
$(pk, sk)$ pair on their local computer without the need to 
verify that identity with a third-party authority.} 
%The theoretical definition of 
Proof-of-work (PoW) is rigorously defined in previous work~\cite{garay15,GKL17,GKL20,PSS17,Ren19,ZTLWLX20} based on the use of the \emph{random oracle model}. 

\begin{definition}[Random Oracle Model]\label{def:rom}
A random oracle $\orcl: \{0, 1\}^* \rightarrow 
\{0, 1\}^\secp$ on input $x \in \{0, 1\}^*$ outputs a value selected
uniformly at random from $\{0, 1\}^\secp$ if $x$ has never been
queried before. Otherwise, it returns the previous value returned when $x$ was queried last. 
\end{definition}
% 
% 
% A \emph{proof-of-work} (PoW) is constructed using the 
% random oracle (Definition~\ref{def:rom}). There are various types of proofs-of-work; here we define
% the most commonly used by previous work~\cite{garay15,GKL17,GKL20,KRS18,PSS17,Ren19,ZTLWLX20}.
% Importantly, proof-of-work is defined with respect to a 
% difficulty parameter $D$ indicating that at least 
% $f(D)$ timesteps (for some polynomial function $f$) must
% be spent by the user of the PoW before they obtain a desired output
% value. \ls{Check this! It is not the well-adopted one. }
% 
\begin{definition}[Bitcoin PoW]\label{def:pow}
All nodes access a common random oracle $\orcl: \{0, 1\}^* \rightarrow \{0, 1\}^\secp$.
% Provided a difficulty parameter $D$ and security parameter $\secp$
% (where $D \leq 2^\secp$),
We say a node successfully performs a PoW with \emph{proof} $x\in \{0, 1\}^*$ if $\orcl(x) \leq D$.
\end{definition}
% 
% Equivalent to using the difficulty parameter $D$, one can instead
% consider $p\triangleq D/2^{\lambda}$. 
% %, the probability of mining a block within some fixed amount of time (usually in one round). 
% The notion of $p$ used in lieu of $D$ has
% been considered in previous work~\cite{garay15,GKL17,GKL20,KRS18,PSS17,Ren19}. 
% %and has been used to simplify the analyses. 
% % Throughout 
% % the rest of this paper, we use $p$ to indicate the probability
% % of mining a block during any given round. The difficulty $D$
% % can be adjusted to obtain any desired $p$ value.

% Given these definitions, we define the final part of the Bitcoin
% protocol which is determining whether a chain is \emph{valid}.

\begin{definition}[Valid Chain]\label{def:validity}
A blockchain $\calC=B_0B_1\cdots B_{\ell} = B_0\langle s_1, x_1, \nce_1\rangle \cdots \langle s_\ell, x_\ell, \nce_\ell\rangle$ is \emph{valid} with respect to a given puzzle difficulty level $D\in \{1, \cdots, 2^{\lambda}\}$ if the following hold: 
(1) $\orcl(B_0) = s_1$ and $\orcl\pth{B_{\ell^{\prime}}} = s_{\ell^{\prime}+1}$ for $\ell^{\prime} = 1, \cdots, \ell-1$; and 
(2) $\orcl\pth{B_{\ell^{\prime}}} \le D$ for $\ell^{\prime} = 0, \cdots, \ell$. 
% \begin{enumerate}
% \item $\orcl\pth{B_{\ell^{\prime}}} = s_{\ell^{\prime}+1}$ ~~~ for $\ell^{\prime} = 1, \cdots, \ell-1$; 
% \item $\orcl\pth{B_{\ell^{\prime}}} \le D$ ~~~ for $\ell^{\prime} = 1, \cdots, \ell$; 
% %    \item $\ctr_{\ell^{\prime}} \le q$ ~~~ for  $\ell^{\prime} = 1, \cdots, \ell$. 
% \end{enumerate}
\end{definition}

\noindent{\bf Longest Chain Rule}\\
The length of a valid chain $C$ is the number of blocks it contains.  We refer to the local version of the blockchain kept by node $i$ as the local chain at node $i$, denoted by $\calC_i$. 
In each round $r$, node $i$ tries to mine a block via solving a PoW puzzle with the specified difficulty $D$. If a block is successfully mined, then node $i$ extends its local chain with this block and broadcasts its updated local chain to all other nodes in the network, which will be delivered at each node at the beginning of the next round. At the beginning of the next round, before working on PoW, node $i$ updates its local chain to be the longest chain it has seen. If there are many longest chains, node $i$ chooses one of them uniformly at random. 

For ease of exposition, henceforth, $C_i$ is referred to the local chain at the end of a round; $C_i(t)$ is the local chain of node $i$ at the end of round $t$. 
% 
% If a node has more than one longest chain, then it breaks 
% ties using some \emph{symmetry-breaking} strategy.
% Since the focus of our paper is on symmetry-breaking,
% we define the following simplified model of the Bitcoin
% cryptosystem in order to focus our attention on symmetry-breaking.
% 
Equivalent to using the difficulty parameter $D$, one can instead
consider $p\triangleq D/2^{\lambda}$. 
The notion of $p$ used in lieu of $D$ has
been considered in~\cite{garay15,GKL17,GKL20,KRS18,PSS17,Ren19} to simplify notation. Henceforth, we will quantify the algorithm performance in terms of $p$ rather than $D$ and $\lambda$.   

We use the phrase \emph{with overwhelming probability} throughout this paper. \emph{With
overwhelming probability} is defined as with probability 
at least $1 - \frac{1}{\poly(\secp)^c}$ for any constant $c \geq 1$.
We use the phrase \emph{with all but negligible 
probability in $\secp$} to mean that the probability
is upper bounded by some negligible function $\nu(\secp)$ 
on $\secp$ (defined in Definition~\ref{def:negl-prob}).

%\ls{Move to right after Def.7? Or point to this def right after Def.7.}
\begin{definition}[Negligible Probability]\label{def:negl-prob}
A function $\nu$ is \emph{negligible} if for every polynomial $p(\cdot)$, there exists an $N$ such that for all integers $n > N$, it holds that $\nu(n) < \frac{1}{p(n)}$. We denote such a function by $\negl$. An event that occurs
with \emph{negligible probability} occurs with probability $\negl(n)$.
\end{definition}
% \begin{definition}[Bitcoin Cryptosystem Model]\label{def:bitcoin-system}
% The Bitcoin model we consider follows the following set of rules:
% \begin{enumerate}
%     \item Each node (honest and Byzantine) mines a block with probability $p$
%     during any round $r$. 
%     Each node can mine at most one block in a round.
%     \label{item:at-most-one}
%     \item Each honest node $v_i$ maintains a local chain $C_i$ where the first
%     block of the chain is the pre-defined \emph{genesis block}.
%     \item At the beginning of each round, honest node $v_i$ chooses the longest \emph{valid} chain it received and set this chain as its local chain.  
%     \item A honest node $v_i$ which mines a block $B$ adds it to its local
%     chain.
% \end{enumerate}
% \end{definition}

\subsubsection{Properties of the Protocol}

In this paper, we will analyze the Nakamoto consensus in terms of two characteristics (generalized
from definitions
in~\cite{garay15,KRS18,ZTLWLX20}). 
The \emph{common prefix} is defined as a sub-chain that is a common prefix of  
% the set of blocks shared by 
the local chains of all honest nodes at the end of a round.
% \ls{QQ, please check the writing in this subsection. } 
The two properties 
\emph{maximal common prefix} and \emph{maximal 
inconsistency} 
%and \emph{consistency} 
are defined
intuitively as: the maximal prefix that 
is the same across all honest chains and 
the maximal number of blocks in any honest chain that is not
shared by all other honest chains, respectively. 
%Note
%that our definition for \emph{common prefix growth} is
%different from definitions present in previous work
%but we find that it is a more intuitive and precise
%definition for our purposes.
% \ls{QQ, please check notation consistency!!!}
\begin{property}[Maximal common-prefix and maximal inconsistency]
\label{def: inconsitent chain tips}
Given a collection of chains $\calC=\sth{\tilde{C}_1, \cdots, \tilde{C}_m}$ that are kept by honest nodes, the maximal common-prefix of chain set $\calC$, denoted by $P_{\calC}$, is defined as the longest common-prefix of chains $\tilde{C}_1, \cdots, \tilde{C}_m$. 
The maximal inconsistency  of $\calC$, denoted by $I_{\calC}$, is defined as 
\begin{align}
\label{def: inconsistency}
\max_{i: 1\le i \le m} \abth{\tilde{C}_i - P_{\calC}},
\end{align}
where $\tilde{C}_i - P_{\calC}$ is the sub-chain of $\tilde{C}_i$ after removing the prefix $P_{\calC}$ and $\abth{\cdot}$ denotes the length of the chain, i.e., the number of blocks in the chain.  
\end{property}

\section{Fundamental Limitations of Existing Approaches} 
\label{sec: fundamental limitation}
To the best of our knowledge, existing work assumes extremely small $p$. 
In fact, the seemingly mild {\em honest majority assumption} in \cite{garay2015bitcoin,pass2017analysis} also implicitly assumes small $p$. 

\begin{proposition}
\label{prop: honest majory necessity}
If the {\em honest majority assumption} in \cite{garay2015bitcoin} holds, then %it must be true that 
$p \le \frac{n-2b}{2(n-b)^2}$. 
\end{proposition}
A formal statement of the honest majority assumption and the proof of Proposition \ref{prop: honest majory necessity}  
can be found in 
\iflong
Appendix \ref{sec: honest maj}. 
\else
\fl.
\fi
Note that the upper bound in this proposition is only a necessary condition. Having $p$ satisfy this condition does not guarantee protocol correctness.   
% 
%\ls{QQ to do: check the consistency of using protocol v.s. algorithm}\qq{I checked it, algorithm is only
%used three times in the entire paper at the appropriate places.}
\begin{remark}\label{rmk: existing slow}
Proposition \ref{prop: honest majory necessity} implies that in the vanilla Nakamoto consensus protocol, 
%for the Bitcoin backbone to work with the probably guarantee assured in \cite{garay2015bitcoin}, 
unless $\frac{b}{n}$ is {\em non-trivially} bounded above from $\frac{1}{2}$, $p$ needs to be extremely low -- even much lower than the commonly believed $\Theta(\frac{1}{n})$. See 
\iflong
Appendix \ref{sec: honest maj}
\else
\fl.
\fi
for detailed arguments. 
\end{remark}
% 
%\qq{$\Delta$ is used differently here than in the intro.} 
To the best of our knowledge, most of the existing analyses focus on bounding the number of ``convergence opportunities'', which for $\Delta=1$ is defined as the number of rounds in which {\em exactly } one honest node mines a block, and for general $\Delta$, it is defined as the global block mining pattern that consists of (i) a period of $\Delta$ rounds where no honest node mines a block, (ii) followed by a round where a single honest player mines a block, (iii) and, finally, another $\Delta$ rounds of silence from the honest nodes \cite{PSS17,KRS18}.  Obviously, guaranteeing sufficiently many convergence opportunities necessarily requires $p$ to be small; in the extreme case when $p=1$ there will be no convergence opportunities at all. 
% \nn{I suggest removing the question and just going into the next paragraph} The following question arises naturally:
% \vskip 0.2\baselineskip
% {
%     \em 
%     Is ``convergence opportunity'' the right quantity to characterize for analyzing correctness and liveness of {\em proof-of-work} protocols?  
%     }
% % 
% % 
% 
An important insight from our results is that \emph{convergence opportunities are not necessary for common-prefix growth}. 
This is illustrated Fig.~\ref{fig:AD} which depicts the chain growth when there are 4 honest nodes and $p=1$. Each node mines a block every round and each is associated with a color. 
\begin{figure}
    \centering
    \includegraphics[width=10cm]{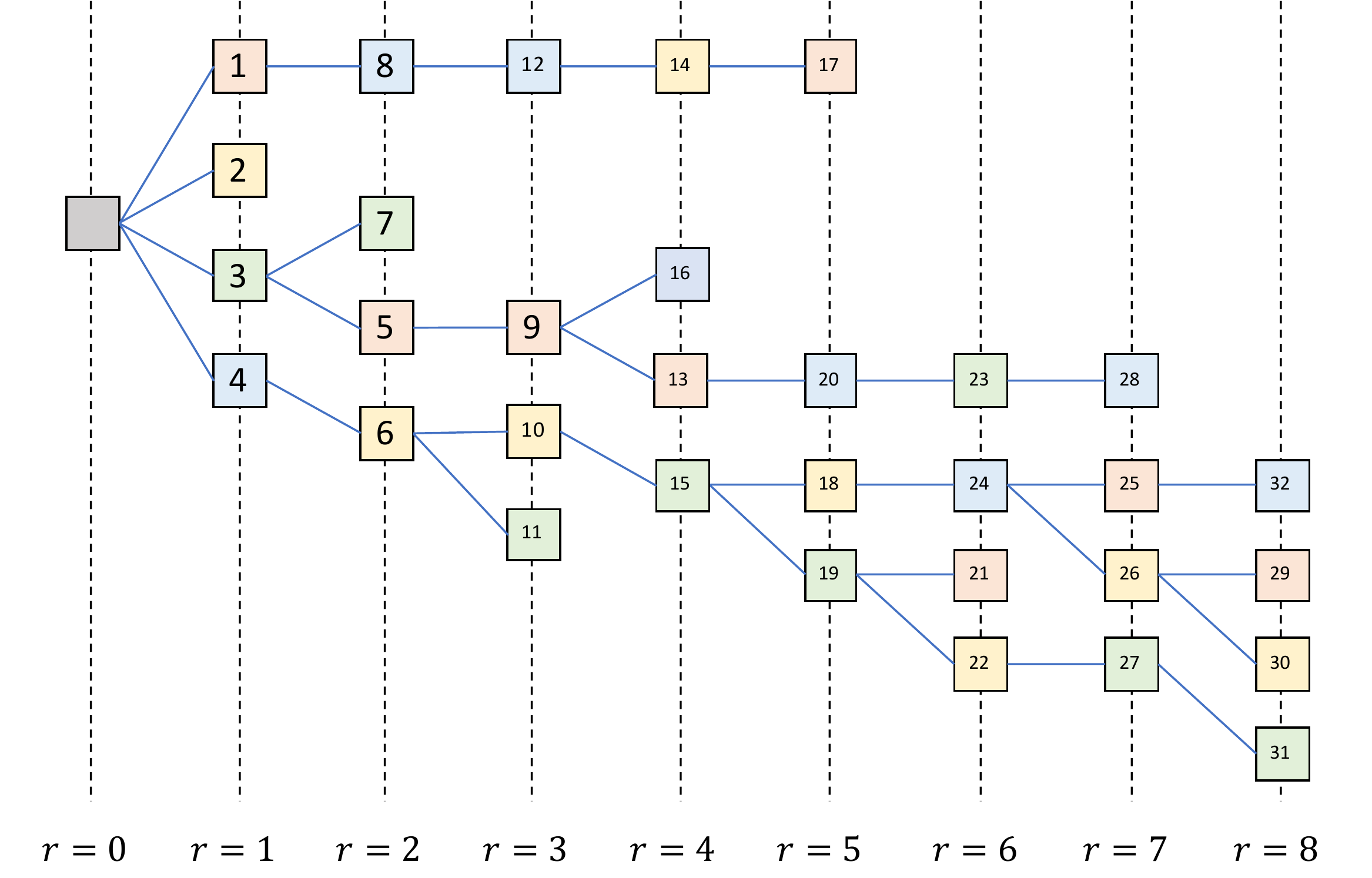}
    \caption{Example growth of a set of chains starting with
    the genesis block at round $r = 0$. Here, in this example
    $p = 1$, $n=4$, and $b=0$.}
    \label{fig:AD}
\end{figure}
In particular, blocks $1, 5, 9, 13, 17, 21, 25, 29$ are mined by the pink node, blocks $4, 8, 12, 16, 20, 24, 28, 32$ are mined by the blue node, etc. In each round, each node chooses one of the existing longest chains uniformly at random to extend. 
As shown in Fig.~\ref{fig:AD}, there are no convergence opportunities in any of these 8 rounds and the four nodes never choose the same chain to extend. However, instead of the trivial common prefix (the genesis block) the longest chains at the end of round 8 (the four chains ending with blocks 32, 29, 30, and 31, respectively) share the common prefix 
$\text{\em genesis} \to 4 \to 6 \to 10 \to 15. $
In general, as we show in Section \ref{sec: main results}, even for the extreme case when $p=1$, the common prefix of the longest chains still grows as time goes by. 
% \nn{By what rate?}
% \ls{Add Anne's simulation results here. Need to acknowledge Anne.}\nn{Added her github \cite{anne}}

\section{Uniformly-at-Random Symmetry-Breaking Strategy} \label{sec: main results}
Bitcoin uses the \emph{first-seen} symmetry-breaking strategy; nodes will only switch to a new chain with more proof-of-work than their current longest chain. 
%Other symmetry-breaking strategies will be discussed in Section \ref{sec:discussion}. 
%include \emph{lexicographically-first} and \emph{global random coin}, which we analyze in Appendix~\ref{sec:symmetry-breaking}. 
%
In this section, we investigate the power of the uniformly-at-random symmetry-breaking strategy, in which each honest node chooses one of its received longest chains uniformly at random to extend upon -- independently of other nodes and independently across rounds. 
%Among a wide variety of alternatives, 
We choose to start with the uniformly-at-random strategy because (1) it is easy to implement, especially in a distributed fashion, and (2) despite its simplicity, it is very powerful in fostering chain growth. 
% 
%Since the analyses on both the adversary-free and adversary-prone settings are highly non-trivial, 

For ease of exposition, we first present our results in the adversary-free setting (Sections \ref{sec: warmup} and \ref{subsec: general p and Adversary-free}) and then in the adversary-prone setting (Section \ref{subsec: general p and Adversary-prone}). 
%Finally, we discuss the implications of our results under different parameter regions. 
%
%The following definitions will be used in the remainder of this section. 

\subsection{Warmup: $p = 1$ and Adversary-Free} \label{sec: warmup}
Even the adversary-free setting (i.e., $b=0$) is surprisingly non-trivial to analyze. Hence we build insights by first considering the simpler setting where $p=1$ as a warmup.

\begin{theorem}
\label{lm: p=1 inde random}
Suppose that $p=1$ and $b=0$. %, i.e., every node can successfully mine a block in each round and the system contains no Byzantine nodes. 
Then for any given round index $t\ge 1$, in expectation, the local chains at the honest nodes share a common prefix of length $t+1-O(n)$. 
\end{theorem}
%\nn{This theorem should have implications for consensus with honest nodes. Do we want to make that point here? For example, might I want to run this protocol instead of Paxos if I can tolerate the delay (Paxos is probably not the right example since this only provides consistency whp, and Paxos can tolerate crashes. But whatever the equivalent is in distributed algorithms)? Why or why not?}
% 
% 
\begin{remark}
In Theorem \ref{lm: p=1 inde random}, the expectation is taken w.\,r.\,t.\,the randomness in the symmetry breaking strategy. 
Theorem \ref{lm: p=1 inde random} says that %the potential serious forking among local chains at different honest nodes can be effectively resolved by introducing the simple random symmetry-breaking. 
%More precisely, 
large $p$ indeed boosts the growth of the common prefix among the local chains kept by the honest nodes, and that, though temporal forking exists among local chains kept by the honest nodes, such forking can be quickly resolved by repetitive symmetry-breaking across rounds. 
\end{remark}

The following definition and theorem are useful to see the intuitions of Theorem \ref{lm: p=1 inde random}. 
\begin{definition}[Coalescing Random Walks \cite{aldous2002reversible}\footnote{The original definition given in \cite{aldous2002reversible} assumes no self-loops, but its analysis applies to the graphs with self-loops.}]
\label{def:coalescing-random-walks}
In a coalescing random walk, a set of particles make independent random walks on a undirected graph $G = (V, E)$ with self-loops. 
Whenever one or more particles meet at a vertex, they unite to form a single particle, which then 
continues the random walk through the graph. 
%Let $C_G$ be the expected time for all $n$ particles to coalesce, when initially one particle is located at each vertex of an $n$ vertex graph.
We define the \emph{coalescence time}, denoted by $C_{G}$,   
to be the number of steps
required before all particles merge into one particle. 
\end{definition}
\begin{theorem}[\cite{aldous2002reversible} \cite{cooper2010multiple}]
\label{thm: coalsecing rw complete}
If $G = (V, E)$ is complete, then $\expect{C_{G}} = O(n)$. 
\end{theorem}
In the proof of Theorem \ref{lm: p=1 inde random}, we build up the connection between the longest chains and the backwards coalescing random walks on complete graphs, 
and show that the maximal inconsistency among $n$ longest chains turns out to be the same as the number of steps it takes $n$ random walks on the $n$-complete graph to coalesce into one. 
Finally, we use the existing results on coalescing random walks to conclude. 
%Then we show that common-prefix length of the longest chains can be bounded by the coalescence time of random walks on complete graphs. 
% Building on the insights obtained in Section \ref{sec: warmup}, we present our results for general $p\in (0, 1)$ in Section \ref{subsec: general p and Adversary-free}. 

\vskip 0.5\baselineskip
\noindent{\bf Main proof ideas of Theorem \ref{lm: p=1 inde random}.}
%Before diving into a formal proof, 
We cast our proof insights via an example presented in Fig.\,\ref{fig:AD}. 
In this figure, there are four miners. For ease of exposition, we use the colors {\em pink}, {\em yellow}, {\em green}, and {\em blue} 
%{\color{pink} \bf pink}, {\color{yellow} \bf yellow}, {\color{green}\bf  green}, and {\color{blue}\bf blue} 
to represent each of the miners, respectively. 
% Recall that when $p=1$, in each round, each miner mines a block. In Fig.\,\ref{fig:AD}, blocks 1, 5, 9, 13, 17, 21, 25, 29 are the blocks mined by the red miner in rounds 1, 2, ..., 8, respectively. Similarly, blocks 2, 6, 10, ..., 30 are the blocks mined by the yellow miner, blocks 3, 7, 11, 15, ..., 27, 31 are the blocks mined by the green miner, and blocks 4, 8, 12, 16, 20, ..., 32 are the blocks mined by the blue miner. 
% 
% \vskip 0.3\baselineskip
% Back to the concrete example in Fig. \ref{fig:AD}. 
As shown in Fig. \ref{fig:AD}, there are 4 longest chains at the end of round 8 and these chains share a maximal common prefix ending at block 15. 
The maximal inconsistency of these 4 longest chains is 4; that is, these 4 longest chains are NOT inconsistent with each other until the most recent 4 blocks of each chain. 
For expository convenience below, instead of using numbers to represent each of the blocks, we use the tuple $\pth{\text{color}, r}$ to represent a block that is mined by a certain miner at round $r$. %For example, $\pth{{\color{blue}\bf blue}, 8}$ corresponds to block 32 in Fig. \ref{fig:AD}. 
%Next we illustrate that 
The maximal inconsistency of the longest chains can be characterized by the coalescing time on complete graphs. 
To see this, let's consider the four longest chains held
by honest miners during
round 8 %and all previous blocks in the chain 
backwards.  

\vskip 0.2\baselineskip
\noindent{Backwards-Chain $\#1$:} 
 $\pth{{\color{blue}\bf blue}, 8} \to \pth{{\color{pink}\bf pink}, 7} \to \pth{{\color{blue}\bf blue}, 6} \to \pth{{\color{yellow}\bf yellow}, 5} \to \pth{{\color{green}\bf green}, 4} \to \pth{{\color{yellow}\bf yellow}, 3} \to \pth{{\color{yellow}\bf yellow}, 2} \to \pth{{\color{blue}\bf blue}, 1} \to \pth{{\color{gray}\bf gray}, 0}$, 
% 
% \begin{align*}
% & \pth{{\color{blue}\bf blue}, 8} \to \pth{{\color{pink}\bf pink}, 7} \to \pth{{\color{blue}\bf blue}, 6} \to \pth{{\color{yellow}\bf yellow}, 5} \to \pth{{\color{green}\bf green}, 4} \to \pth{{\color{yellow}\bf yellow}, 3}  \\
% & \quad \to \pth{{\color{yellow}\bf yellow}, 2} \to \pth{{\color{blue}\bf blue}, 1} \to \pth{{\color{gray}\bf gray}, 0}, 
% \end{align*}
which can be read as ``block $\pth{{\color{blue}\bf blue}, 8}$ is attached to block $\pth{{\color{pink}\bf pink}, 7}$ which is further attached to block $\pth{{\color{blue}\bf blue}, 6}$ ... attached to the genesis block $\pth{{\color{gray}\bf gray}, 0}$. ''

\vskip 0.2\baselineskip
\noindent{Backwards-Chain $\#2$}:  $\pth{{\color{pink}\bf pink}, 8}  
\to \pth{{\color{yellow}\bf yellow}, 7}
\to \pth{{\color{pink}\bf pink}, 6}
\to \pth{{\color{green}\bf green}, 5}
\to \pth{{\color{green}\bf green}, 4} 
\to \pth{{\color{yellow}\bf yellow}, 3} \to \pth{{\color{yellow}\bf yellow}, 2} 
\to \pth{{\color{blue}\bf blue}, 1}
\to \pth{{\color{gray}\bf gray}, 0}.$ \\
% 
% \begin{align*}
% &   \pth{{\color{pink}\bf pink}, 8}  
% \to \pth{{\color{yellow}\bf yellow}, 7}
% \to \pth{{\color{pink}\bf pink}, 6}
% \to \pth{{\color{green}\bf green}, 5}
% \to \pth{{\color{green}\bf green}, 4} 
% \to \pth{{\color{yellow}\bf yellow}, 3} \\
% & \quad \to \pth{{\color{yellow}\bf yellow}, 2} 
% \to \pth{{\color{blue}\bf blue}, 1}
% \to \pth{{\color{gray}\bf gray}, 0}, 
% \end{align*}
%which can be read as ``block $\pth{{\color{pink}\bf pink}, 8}$ is attached to block $\pth{{\color{yellow}\bf yellow}, 7}$ which is further attached to block $\pth{{\color{pink}\bf pink}, 6}$... attached to the genesis block $\pth{{\color{gray}\bf gray}, 0}$. '' Similarly, we have two more backwards-chains. \\
\noindent{Backwards-Chain $\#3$}: $\pth{{\color{yellow}\bf yellow}, 8}  
\to \pth{{\color{yellow}\bf yellow}, 7}
\to \pth{{\color{pink}\bf pink}, 6}
\to \pth{{\color{green}\bf green}, 5}
\to \pth{{\color{green}\bf green}, 4} 
\to \pth{{\color{yellow}\bf yellow}, 3} 
\to \pth{{\color{yellow}\bf yellow}, 2} 
\to \pth{{\color{blue}\bf blue}, 1}
\to \pth{{\color{gray}\bf gray}, 0}$. 
% \begin{align*}
% &   \pth{{\color{yellow}\bf yellow}, 8}  
% \to \pth{{\color{yellow}\bf yellow}, 7}
% \to \pth{{\color{pink}\bf pink}, 6}
% \to \pth{{\color{green}\bf green}, 5}
% \to \pth{{\color{green}\bf green}, 4} 
% \to \pth{{\color{yellow}\bf yellow}, 3} \\
% & \quad \to \pth{{\color{yellow}\bf yellow}, 2} 
% \to \pth{{\color{blue}\bf blue}, 1}
% \to \pth{{\color{gray}\bf gray}, 0}. 
% \end{align*}

\noindent{Backwards-Chain $\#4$}:  $\pth{{\color{green}\bf green}, 8}  
\to \pth{{\color{green}\bf green}, 7}
\to \pth{{\color{yellow}\bf yellow}, 6}
\to \pth{{\color{green}\bf green}, 5}
\to \pth{{\color{green}\bf green}, 4} 
\to \pth{{\color{yellow}\bf yellow}, 3}
\to \pth{{\color{yellow}\bf yellow}, 2} 
\to \pth{{\color{blue}\bf blue}, 1}
\to \pth{{\color{gray}\bf gray}, 0}. $
% \begin{align*}
% &   \pth{{\color{green}\bf green}, 8}  
% \to \pth{{\color{green}\bf green}, 7}
% \to \pth{{\color{yellow}\bf yellow}, 6}
% \to \pth{{\color{green}\bf green}, 5}
% \to \pth{{\color{green}\bf green}, 4} 
% \to \pth{{\color{yellow}\bf yellow}, 3} \\
% & \quad \to \pth{{\color{yellow}\bf yellow}, 2} 
% \to \pth{{\color{blue}\bf blue}, 1}
% \to \pth{{\color{gray}\bf gray}, 0}. 
% \end{align*}
% 

\noindent Since $p=1$ and there is no adversary, the number of longest chains received by each honest node at each round is $n$. 
Under our symmetry-breaking rule, in each round $t$, each miner chooses which of the longest chains received at the beginning of round $t$ to extend on uniformly-at-random. Thus, neither the previous history up to round $t$ nor the future block attachment choices after round $t$ affects the choice of the chain extension in round $t$. 
Reasoning heuristically\footnote{Formally shown in the proof of Theorem \ref{lm: p=1 inde random} via introducing an auxiliary process.},  we can view each of the backwards-chain as a random walk on a $4$-complete graph with vertex set $\sth{pink, yellow, green, blue}$. In particular, Backwards-Chain $\#1$ can be viewed as a sample path of a random walk starting at the blue vertex, then moves to the pink vertex, then back to the blue vertex etc., and finally to the blue vertex. Similarly, Backwards-Chains $\#2, \#3$, and $\#4$ can be viewed as the sample paths of three random walks starting at the pink vertex, yellow vertex, and green vertex, respectively. 
These four random walks (starting at four different vertices) are not completely independent. For any pair of random walks, before they meet, they move on the graph independently of each other; whenever they meet, they move together henceforth. Concretely, backwards-chains 2 and 3 meet at $\pth{{\color{yellow}\bf yellow}, 7}$ and these chains are identical starting from block $\pth{{\color{yellow}\bf yellow}, 7}$; this holds similarly for other pairs of backwards chains. Finally, these four backward chains all meet at the block $\pth{{\color{green}\bf green}, 4}$ and move together henceforth. Notably, this block is exactly the last block in the maximal common prefix of the four longest chains of round 8. Thus, the maximal inconsistency among the longest chains of round 8 is identical to the number of backwards steps it takes for all these four random walks to coalesce into one. This relation is not a coincidence. It can be shown (detailed in the proof of Theorem \ref{lm: p=1 inde random}) that this identity holds for general $n$. % and for general sample paths of the chain realizations and the random walks. 
Formal proof of Theorem \ref{lm: p=1 inde random} can be found in Appendix \ref{app: proof of p =1}.

\subsection{General p: Adversary-Free}
\label{subsec: general p and Adversary-free}
The analysis for general $p$ is significantly more challenging than that of $p=1$ in two ways: (1) we need to repeatedly apply coupling arguments; and (2) we need to characterize the coalescence time of a new notion of coalescing random walks (the lazy coalescing random walks), the latter of which could be of independent interest for a broader audience.  
% 
% For both cases of $p=1$ and $p\in (0, 1)$, we first build up the connection between the longest chains and the backwards coalescing random walks on complete graphs. 
% Then we show that common-prefix length of the longest chains can be bounded by the coalescence time of random walks on complete graphs. 

\begin{theorem}
\label{lm: general p inde random}
Suppose that $np = \Omega(1)$. 
If $p<\frac{4\ln 2}{n}$, in expectation, at the end of round $t$, the local chains at the nodes share a common prefix of length $\pth{1+\pth{1-\pth{1-p}^n}t}-O(\frac{1}{npe^{-np}})$. 
If $p\ge \frac{4\ln 2}{n}$, in expectation, at the end of round $t$, the local chains at the nodes share a common prefix of length $\pth{1+\pth{1-\pth{1-p}^n}t}-O\pth{\frac{2np}{\pth{1 - 2 \exp\pth{-\frac{1}{3}np}}}}$. 
\end{theorem}
\begin{remark}
The expression of the common prefix length in Theorem \ref{lm: general p inde random} contains two terms 
with the first term (i.e., $\pth{1+\pth{1-\pth{1-p}^n}t}$) being the only term that involves $t$. Intuitively, from this term, we can read out the common prefix length growth rate w.r.t. $t$.  
The second term (which is expression in terms of Big-O notation) can be interpreted as a quantification of the maximal inconsistency of the honest chains.

Now we further interpret these two terms 
%common prefix length growth rate  
via simplifying the expression using the inequalities $\pth{1 - np} \le (1-p)^n \le \exp\pth{-np}$. 
%The following asymptotic notions are asymptotic in $n$ and $t$.  

\noindent (1) When $np = o(1)$, it is true that $(1-p)^n \approx (1-np)$ for large $n$, which implies that $\pth{1-\pth{1-p}^n}t \approx np t = o(t)$, i.e., the common prefix grows at a speed $o(t)$. The maximal inconsistency bound $O(\frac{1}{npe^{-np}})$ is not tight. 
Nevertheless, via a straightforward calculation, we know that the maximal inconsistency is $O(1)$. \\
(3) When $np=\omega(1)$,  we have $0\le (1-p)^n \le \exp\pth{-np} \to 0$ as $np\to \infty$. Thus  the common-prefix grows at the speed 
$\pth{1-\pth{1-p}^n}t \approx t = \Omega(t)$ with maximal inconsistency $O(np)$ for sufficiently large $np$. \\
(4) When $np = c\in (0, 1)$, it is true that $(1-p)^n = \pth{1-c/n}^n \to \exp\pth{-c}$ as $n\to \infty$. The common-prefix grows at the speed of $\Theta(t)$ for sufficiently large $n$ and the maximal inconsistency is $O(1)$. 

Overall, when $np$ gets larger, the common-prefix growth increases and the maximal inconsistency grows at a much slower rate. %is not large at all.    
\end{remark}

The following definition and lemma are used in proving  Theorem \ref{lm: general p inde random}. This lemma could be of independent interest to a broader audience and its proof can be found in the appendix. 
% \iflong
% Appendix \ref{app subsec: general p and Adversary-free}. 
% \else
% \fl.
% \fi

\begin{definition}[Lazy coalescing random walk]
\label{def: lazy-coalescing-random-walk}
For any fixed $u\in (0,1)$, we say $n$ \emph{particles} are $u$-lazy coalescing random walks if 
for each step: 
with probability $(1-u)$, each %of the remaining particles
particle stays at its current location; 
with probability $u$, each particle moves to 
an adjacent vertex picked uniformly at random. 
If two or more particles meet at a location, they unite into a single particle and continue the procedure. 
The \emph{coalescence time} is the same as that in Definition \ref{def:coalescing-random-walks}.
% Whenever one or more particles meet at a vertex, they unite to form a single particle, which then 
% continues the random walk through the graph. 
% %Let $C_G$ be the expected time for all $n$ particles to coalesce, when initially one particle is located at each vertex of an $n$ vertex graph.
% We define the \emph{coalescence time}, denoted by $C_{G}$,   
% to be the number of steps
% required before all particles merge into one particle. 
\end{definition}

\begin{lemma}
\label{lm: lazy-coalescence-time}
Suppose that $G$ is a complete graph of size $|V|=n_g$ (where $n_g\ge 2$) with self-loops. For any $u\in (0,1)$, the coalescence time of the $u$-lazy coalescing random walks is $C_{G}(n_g) = O(n_g/u)$.  
% \begin{align}
% \label{eq:lazy-coalescence-time}
% C_{G}(n_g) = O(n_g/u). 
% \end{align}
\end{lemma}

%This result is of independent interest to a broader audience. 

\noindent{\bf Proof Sketch of Theorem \ref{lm: general p inde random}. }
When $p< \frac{4\ln 2}{n}$, we can use Poisson approximation to approximate the distribution of number of blocks in each round. A straightforward calculation shows that the probability of having exactly one block in a round is $np \exp\pth{-np}$. Thus, in expectation, the maximal inconsistency is $O\pth{\frac{1}{np \exp\pth{-np}}}$. 
Henceforth, we restrict our attention to the setting where $p\ge \frac{4\ln 2}{n}$ and quantify the expected maximal inconsistency among the longest chains of round $t$.  
It is attempting to apply arguments similar to that in the proof of Theorem \ref{lm: p=1 inde random} and derive a bound on the maximal inconsistency via stochastic dominance. However, the obtained bound on the maximal inconsistency is $O(n)$ which could be extremely loose for a wide range of $p$. 
Nevertheless, based on the insights obtained in this coarse analysis, we can come up with a much finer-grained analysis and obtain the bound in Theorem \ref{lm: general p inde random}. 
Similar to the proof of the special case when $p=1$, in our fine-grained analysis for general $p\in (0, 1)$, we couple the growth of the common prefix in Nakamoto protocols with the coalescing time random walks on complete graphs. The major differences from the proof of $p=1$ are: (1) instead of the standard coalescing random walks, we need to work with a lazy version of it, formally defined in Definition \ref{def: lazy-coalescing-random-walk}; (2) there is no fixed correspondence between a color and a node -- in our proof of general $p$, the correspondence is round-specific rather than fixed throughout the entire dynamics; (3) there is no bijection between a sample path of the Nakamoto dynamics and that of the backwards coalescing random walks, thus, we need to rely on stochastic dominance to build up the connection of these two dynamics. 
%Fortunately, though the proof for $p\in (0, 1)$ is significantly more involved, the high-level ideas in the proof of the special case $p=1$ applies.  

\subsection{General p: Adversary-Prone}
\label{subsec: general p and Adversary-prone}
Throughout this section, we assume $p<1$. 
In this subsection, we consider adversary-prone systems, i.e., $b>0$. 
Simple concentration arguments show that 
when $bp\ge (1+2c)$ for any given $c\in (0,1)$, using vanilla Nakamoto consensus the chain quality could be near zero.  
To make larger $p$ feasible, we introduce a new assumption---Assumption~\ref{prop:adversary-advantage}---which we  then remove  in~\cref{sec:vdf-protocol} by providing a construction that ensures Assumption~\ref{prop:adversary-advantage} with all but negligible
probability.
Specifically,
we use a cryptographic tool called a VDF to ensure that over a sufficiently long time window, the corrupt nodes can only collectively extend a chain by more than one block in a round with negligible probability.

\begin{assumption}\label{prop:adversary-advantage}
In each round, a chain can be extended by at most 1 block. 
% Given a collection of chains $\calC=\sth{\tilde{C}_1, \cdots, \tilde{C}_m}$
% at the beginning of round $r \geq 1$ (before any blocks
% are added in round $r$), %the following property holds for round $r$: 
% at most one block can be added to any 
% chain in round $r$. 
%In other words, this ensures the  lengths of all chains, $C_i \in \calC$ is at most $r$ (recall that all chains start with the genesis block $B_0$) at the beginning of round $r$. % 
%\nn{is block $B_i$ mined in round $r$? If so I think you should start with round 1, not 0?}. %\qq{$B_0$ is the genesis block which is not mined by anyone. Also, I changed
%this to $r \geq 1$ but the property holds for round
%$r = 0$ too.}
\end{assumption}

% 
% \nn{I'm super confused about this. Why does this need to happen, since all the chains were broadcast at the END of the previous round?}
To strengthen the  protocol robustness,  
we make the additional minor modification requiring each honest node to selectively relay chains at the \emph{beginning} of a round.  
\vskip 0.2\baselineskip
\noindent{\bf Selective relay rule: }
At each honest node $i$, for each iteration $t\ge 1$: Node $i$ looks at the chains it received in the previous round $t-1$, and if any of them are longer than its own local longest chain, it not only chooses one of the longest chains to replace its local one, it also broadcasts it to other nodes before it begins mining in round $t$.

%If the length of its received longest chains at the beginning of iteration $t$ is %longer than its local chain, then, before making any block mining attempt in %iteration $t$, the node $i$ multi-casts \nn{Have we defined multicasts?} one of %the received longest chains to others. 
% 

\vskip 0.5\baselineskip

\noindent  As implied by our proof, this modification can reduce the maximal difference  between the lengths of the longest chains kept by the honest nodes and by the corrupt nodes. Intuitively, if the adversary sends two chains of different lengths to two different groups of honest nodes, with the selective relay rule, only the longer chain would survive in this round. Notably, it is possible that none of them survive in this round. 
Even with the assurance guaranteed by Assumption~\ref{prop:adversary-advantage}, compared with the adversary-free settings, the analysis for the adversary-prone setting is challenging. 
This is because the corrupt nodes could deviate from the specified symmetry breaking rule. For example, a corrupt node can choose not to extend its longest chain, or can choose from its set of longest chains in any way that provides advantage. In addition, a corrupt node can hide blocks it has mined from the honest nodes for as long as it wants, or from some subset of the honest nodes during a round. 
% \nn{Actually it can't hide from a subset, because of the selective relay rule, right?}\ls{Still can hide for a given round. But cannot lie about hide for too many rounds, which is implied by the lemma on adversarial advantage. }

% {\red For simplicity, we assume that eventually an adversary will choose randomly among chains that end in an honest block but for a limited time they can arbitrarily choose. } 
For simplicity and for technical convenience, we assume that a corrupt node randomly chooses among longest chains that end with an honest block. 
This assumption is only imposed in the rare event when simultaneously both the adversary has no adversary advantage (see Definition \ref{def: adversarial ad}) and only honest nodes mine blocks in the most recent nonempty round.
%\nn{TODO: cut this next sentence?} This applies to many real-world scenarios wherein a corrupt node is myopic and cannot count the number of blocks mined by the corrupt nodes in a chain.

In contrast to the adversary-free setting where the lengths of honest nodes' local chains differ by at most 1, in the presence of an adversary, such difference could be large. 
To precisely bound this difference, we introduce a random process we call {\em adversary advantage}:
\begin{definition}[Adversary advantage]
\label{def: adversarial ad}
%Given the number of blocks generated in each round, we construct the following random process which we refer to as the {\em adversary advantages} sequence. 
Let $\sth{\calN(t)}_{t=0}^{\infty}$ be the random process defined as 
\begin{itemize}
\item $\calN(0)=0$, and 
\item for $t\ge 1$, 
\begin{align*}
\calN(t) =   
\begin{cases}
\calN(t-1) + 1, & ~ \text{if only corrupt nodes found blocks in round $t$;} \\
\max\{\calN(t-1) -1, ~ 0\}, & ~  \text{if only honest nodes found blocks in round $t$;} \\
\calN(t-1), & ~ \text{otherwise.} \\
\end{cases}
\end{align*}
\end{itemize}
\end{definition}
Note that the random process $\sth{\calN(t)}_{t=0}^{\infty}$ is independent of the adversarial behaviors of the corrupt nodes. 
To make the discussion concrete, we introduce the following definition. 
\begin{definition}
The length of the longest chains kept by the honest nodes {\bf at round $t$} 
is defined as the length of the longest local chains kept by  honest nodes {\em at the end} of round $t$. 
\end{definition}
\begin{lemma}
\label{lm: new proof}
For any $t\ge 1$, at the end of round $t$, the length of the longest chains kept by the adversary -- henceforth referred to as an adversarial longest chain of round $t$ -- is at most $\calN(t)$ longer than the length of a chain kept by an honest node. %\qq{, with all but negligible probability in $\secp$.}
\end{lemma}
Proof of Lemma \ref{lm: new proof} can be found in 
\iflong
Appendix \ref{app: general p and Adversary-prone}. 
\else
\fl.
\fi
From its proof, we can deduce an attacking strategy of the adversary that meets the upper bound in Lemma \ref{lm: new proof}. The following lazy random walk, referred to as {\em coalescing opportunities}, is important in our analysis. It can also be used to quantify the chain quality. 
\begin{definition}
\label{def: Coalescing opportunities couting}
Let $t_1, t_2, \cdots$ be the rounds in which at least one node mines a block with the understanding that $t_0 =0$. 
Let $\calJ(m)$ be a random walk defined as 
\begin{align*}
\calJ(m) = 
\begin{cases}
0, ~ ~ &\text{if } m=0; \\
\calJ(m-1)+1, ~~ &\text{if only honest nodes mine a block during round } t_k ; \\
\calJ(m-1)-1, ~ ~ &\text{if only corrupt nodes mine a block during round } t_k ; \\
\calJ(m-1), ~ ~ &\text{otherwise.}
\end{cases}
\end{align*}
\end{definition}
\begin{remark}
A couple of interesting facts on the coalescing opportunities dynamics are: 
Among the most recent $m$ blocks in a longest chain, there are at least $\calJ(m)$ blocks mined by the honest nodes.  
In addition, regardless of the behaviors of the adversary, for any two longest chains, there are at least $\calJ(m)$ block positions each of which has non-zero probability of being in the common prefix of these two chains. 
\end{remark}

Let $p_{+1} = \prob{\calJ(m) = \calJ(m-1)+1}$ and $p_{-1} = \prob{\calJ(m) = \calJ(m-1)-1}$, i.e., $p_{+1}$ (resp. $p_{-1}$) is the probability for $\calJ(m)$ to move up (resp. down) by 1. 
We have 
%$$
%Let $E_1$ denote the event that a nonempty block layer contains no corrupt blocks and $E_2$ denote the event that a nonempty block layer contains only corrupt blocks. For any given $m\ge 1$, we have 
\begin{align}
\label{eq: move up and down nonempty}
p_{+1} = \frac{(1-p)^b\pth{1-\pth{1-p}^{n-b}}}{1 - (1-p)^{n}} ~ ~ \text{and} ~ ~ p_{-1} = \frac{\pth{1-(1-p)^b}\pth{1-p}^{n-b}}{1 - (1-p)^{n}}. 
\end{align}
% and 
% \begin{align*}
% \prob{\calJ(t) = \calJ(t-1)}  = 1 - \prob{E_1} - \prob{E_2}  = 1 -  (1-p)^b\pth{1-\pth{1-p}^{n-b}} -  \pth{1-(1-p)^b}\pth{1-p}^{n-b}.  
% \end{align*}
% 
It is easy to see that when $b>\frac{1}{2}n$, it holds that $p_{+1} > p_{-1}$. For ease of exposition, let $p^* = \prob{\calJ(t) \not= \calJ(t-1)} = p_{+1} + p_{-1}. $ 
%With high probability $\calJ(M)$ is bounded from below.  
\begin{lemma}%[Bounding coalescing opportunities]
\label{lm: the number of coalescing opportunities}
With probability at least $\pth{1- \exp\pth{-\frac{(p_{+1} - p_{-1})^2M}{16p^*}} - \exp\pth{-\frac{(p^*)^2 M}{2}}}$, it holds that $\calJ(M) \ge \frac{(p_{+1} - p_{-1})M}{4}$. 
% \begin{align*}
% \calJ(M) \ge \frac{(p_{+1} - p_{-1})M}{4}.      
% \end{align*}
\end{lemma}
Lemma \ref{lm: the number of coalescing opportunities} gives a high probability lower bound on the number of coalescing opportunities during $M$ nonempty rounds. Its proof can be found in Appendix \ref{app: general p and Adversary-prone}. 

\begin{theorem}
\label{thm: general p adversarial}
For any given $T\ge 1$ and $M\ge \frac{4}{\beta(p_{+1} - p_{-1})}$ where $\beta = \frac{(n-b)p}{ 2\pth{3np}^2}$, 
at the end of round $T$, 
with probability at least 
\[
1- \exp\pth{-\frac{(p^*)^2M}{2}} - \exp\pth{-\frac{(p_{+1} - p_{-1})^2M}{16p^*}} - \frac{2}{\beta}\exp\pth{-\frac{1}{2}(n-b)}  
\]
over the randomness in the block mining, the expected maximal inconsistency among a given pair of honest nodes is less than $M$, where the expectation is taken over the randomness in the symmetry breaking. 
\end{theorem}
\begin{remark}
It is worth noting that $\beta = \frac{(n-b)p}{ 2\pth{3np}^2} = \frac{1}{18} \frac{(n-b)}{n}\frac{1}{np}$, i.e., $\beta$ is a function of the fraction of honest nodes and the total mining power of the nodes in the system. 

Suppose that $n\ge 2\log \frac{4}{\epsilon \beta}$ for any given $\epsilon\in (0, 1)$. 
Let 
\[
M^* = \max \sth{\frac{4\log 1/\epsilon}{(p^*)^2},   \frac{4}{\beta(p_{+1} - p_{-1})}, \frac{16p^*}{\pth{p_{+1} - p_{-1}}^2}\log\frac{4}{\epsilon}}. 
\]
From Theorem \ref{thm: general p adversarial}, we know that with probability at least $1-\epsilon$, the maximal inconsistency is less than $M^*$. Roughly speaking, when $b$ gets smaller, $M^*$ mainly gets smaller.  
\end{remark}

\begin{proof}[\bf Proof of Theorem \ref{thm: general p adversarial}]
% 
% 
% The same as that in the proofs of Theorems \ref{lm: p=1 inde random} and \ref{lm: general p inde random}, 
We use $N_t$ to denote the number of blocks generated during round $t$ and associate each node with a distinct color in $\{c_1, \cdots, c_n\}$. 
If node $i$ mines a block during round $t$, we use $\pth{c_i, t}$ to denote this block. The genesis block is denoted as $\pth{c_1, 0}$. 
Recall that the blocks mined during round $t$ are collectively referred to as the block layer $t$.  
As the randomness in the block generation (i.e., puzzle solving of individual nodes) is independent of the adversarial behaviors of the corrupt nodes and is independent of which chain an honest node chooses to extend, we consider the auxiliary process wherein the nodes mine blocks for the first $T$ rounds, and then the corrupt nodes and honest nodes sequentially decide on block attachments. 
Let $\{i_1, \cdots, i_K\}$ be the set of rounds such that $N_{i_k} \not=0$ for each $i_k\in \{i_1, \cdots, i_K\}$.  
Let $j_1$ and $j_2$ be any two honest nodes whose chains at the end of round $T$ are denoted by $C_1(T)$ and $C_2(T)$, respectively. 
For each of these chains, we can read off a sequence of colors
\begin{align*}
\text{for Chain } C_1(T): ~ ~ & c_1 c(1,2) c(1,3)\cdots c(1,\ell_1), ~ ~ \text{and} \\
\text{for Chain } C_2(T): ~ ~ & c_1 c(2,2) c(2,3)\cdots c(2,\ell_2),  
\end{align*}
where $\ell_1$ and $\ell_2$, respectively, are the lengths of chains $C_1(T)$ and $C_2(T)$, $c_1$ is the color of the genesis block, $c(1, k)$ for $k\in \{2, \cdots, \ell_1\}$ is the color of the $k$--th block in $C_1(T)$ and $c(2, k)$ for $k\in \{2, \cdots, \ell_2\}$ is the color of the $k$--th block in $C_2(T)$. 
If $\ell_1\not=\ell_2$, without loss of generality, we consider the case that $\ell_1<\ell_2$; the other case can be handled similarly. We augment the color sequence $c_1 c(1,2) c(1,3)\cdots c(1,\ell_1)$ to the length $\ell_2$ sequence as 
\[
c_1 c(1,2) c(1,3)\cdots c(1,\ell_1) c(1, \ell_1+1) \cdots c(1, \ell_2), 
\]
by setting $c(1, k) = c_0$ for $k=\ell_1+1, \cdots, \ell_2$ where $c_0\notin \{c_1, \cdots, c_n\}$ is a special color that never shows up in a real block. 
It is easy to see that $C_1(T)$ and $C_2(T)$ {\em  start} to be inconsistent at their $k$-th block if and only if $c(1, k^{\prime})\not= c(2, k^{\prime})$ for each $k^{\prime} \in \{k, \cdots, \ell_2\}$. 
% In the remaining proof, we first bound $\pth{\ell_2 - \ell_1}$. 
% Then we bound the maximal consistency number of the $C_1(T)$ and the pruned chain $C_2(T)$, denoted by $\tilde{C}_2(T)$,  whose most recent $(\ell_2-\ell_1)$ blocks are pruned away. 
% 
Let $\{i_{h_1}, \cdots, i_{h_R}\} \subseteq \{i_1, \cdots, i_K\}$ such that for each $i_{h_r}\in \{i_{h_1}, \cdots, i_{h_R}\}$ it holds that 
\begin{itemize}
    \item Only honest nodes successfully mined blocks; 
    \item $\calN(i_{h_r - 1}) =0$. 
\end{itemize}
For ease of exposition, we refer to each of $i_{h_r}$ as {\em a coalescing opportunity}. 
Recall that each of the honest nodes extends one of the longest chains it receives. By Lemma \ref{lm: new proof}, we know that each of $C_1(T)$ and $C_2(T)$ contains a block generated during round $i_{h_r}$. % for $r=1, \cdots, R$. 
Let $\pth{c_1^{\prime}, i_{h_r}}$ and $\pth{c_2^{\prime}, i_{h_r}}$ be the blocks included in $C_1(T)$ and $C_2(T)$, respectively. 
If $\pth{c_1^{\prime}, i_{h_r}}$ is in the $k$-th position in $C_1(T)$, then $\pth{c_2^{\prime}, i_{h_r}}$ is also in the $k$-th position in $C_2(T)$. 
For each $i_{h_r}$, we denote the set of chains (including the forwarded chains) received by $j_1$ and $j_2$ at round $i_{h_r}$, denoted by $\calC_1^r$ and $\calC_2^r$. Since the adversary can hide chains to a selective group of honest nodes, $\calC_1^r$ and $\calC_2^r$ could be different.  
The probability of $j_1$ and $j_2$ extending the same chain at round $i_{h_r}$ is  
\begin{align}
\label{eq: meet prob}
 \frac{\abth{\calC_1^r \cap \calC_2^r}}{\abth{\calC_1^r}\abth{\calC_2^r}} \ge \frac{\text{NB}(i_{h_r-1})}{\pth{\text{NB}(i_{h_r-1}) + \text{AB}(i_{h_r-1}) + \text{AB}(\widetilde{i_{h_r-1}})}^2} %\overset{(b)}{\ge} \frac{\frac{1}{2}(n-b)p}{ \pth{\frac{3}{2}(n-b)p + 3bp}^2}, 
\end{align}

\noindent where the inequality follows from 
\iflong
Lemma \ref{lm: common longest chains}.
\else
Lemma 33 of the \fl.
\fi
% and inequality (b) follows from Hoeffding's inequality. 
By Lemma \ref{lm: the number of coalescing opportunities}, we know that in the $M$ non-empty block layers that are most recent to round $T$, 
\begin{align*}
R\ge \calJ(M) \ge \frac{(p_{+1} - p_{-1})M}{4}
\end{align*}
holds with probability at least $\pth{1- \exp\pth{-\frac{(p^*)^2M}{2}} - \exp\pth{-\frac{(p_{+1} - p_{-1})^2M}{16p^*}}}$. 
%Thus, $R \ge \frac{p^*(p_{+1}^* - p_{-1}^*)M}{4}$ with probability at least $\pth{1- \exp\pth{-\frac{p^*M}{2}} - \exp\pth{-\frac{p^*(p_{+1}^* - p_{-1}^*)M}{4}}}$. 
In addition, it can be shown that for each of the $r$ ensured by Lemma \ref{lm: the number of coalescing opportunities} we have 
\begin{align*}
\max\{\abth{\calC_1^r},  \abth{\calC_2^r}\} \le \text{NB}(i_{h_r-1}) + \text{AB}(i_{h_r-1}) + \text{AB}(\widetilde{i_{h_r-1}}) \indc{\text{AB}(i_{h_r-1})=0}. 
\end{align*}
For any $i_k$, let $X_k$ be the number of blocks mined by the honest nodes during round $i_k$ such that $X_k\not=0$. 
Using conditioning and Hoeffding's inequality, the following holds with probability at least $\pth{1-2\exp\pth{-\frac{1}{2}(n-b)}}$, 
%$\pth{1-2\exp\pth{-\frac{1}{2}(n-b)}}$: 
\[
X_k\ge \frac{1}{2}(n-b)p ~  ~ \text{and } ~  X_k+ Y_k + Y_{k-1}\indc{Y_k=0} \le 3np, % X_k+Y_k + Y_{k-1} \le 3np. 
\]
which implies that $\frac{X_k}{X_k+Y_k+Y_{k-1}\indc{Y_k=0}} \ge \frac{(n-b)p}{ 2\pth{3np}^2} \triangleq \beta$.  
On average over the random symmetry breaking, it takes at most $1/\beta$ coalescing opportunities backwards for chains $C_1(T)$ and $C_2(T)$ to coalesce into one.  
Thus, we need $\frac{(p_{+1} - p_{-1})M}{4} \ge \frac{1}{\beta}$.

\end{proof}

\section{VDF-Based Scheme}
\label{sec:vdf-protocol}
%\nn{The VDF scheme is more than just symmetry-breaking---it requires nodes to apply new validation rules (valid VDF)}
%\qq{I changed the title}
In this section, we present a scheme to ensure Assumption~\ref{prop:adversary-advantage}.  
The key cryptographic tool we use in the following scheme is the 
construction of the \emph{verifiable delay function}, $\func(x)$, which we define informally below. 
Please refer to~\cite{vdf} for the formal definition (also defined formally in the full version of our paper).
%We assume the existence of a \emph{verifiable delay function} as defined below. 
%We denote our VDF function by
%$\func(x)$ where $x$ is the input.
%We provide a short, intuitive definition of 
%verifiable delay functions here. The precise
%definition can be found in Definition~\ref{def:vdfs} in Appendix~\ref{app: additional definitions}.

\begin{definition}[Verifiable Delay Function (informal)]\label{def:vdfs-short}
Let $\secp$ be our security parameter. 
There exists a function $\func$ with 
difficulty $\dif = O(\poly(\secp))$ 
where the output $y \leftarrow
\func(x)$ (where $x \in \{0, 1\}^\secp$) 
cannot be computed
in less than $\dif$ sequential computation steps,
even provided $\poly(\secp)$
parallel processors, with probability at least $1-\negl(\secp)$. The VDF output can be verified, quickly, in $O(\log(\dif))$ time.
\end{definition}

We set the difficulty of the VDF to the duration of a round;
in other words, the difficulty is set such that the 
VDF produces exactly one output at the end of each round.
% In a network model measured by communication delay $\Delta$, the difficulty of the VDF is set to $\Delta$.
% In order to tolerate a large $p$ in the presence of adversaries, we introduce the following VDF-based scheme. 
We amend default Nakamoto consensus by adding the following procedure.
We believe this could be added in a backwards-compatible way to existing Nakamoto implementations, like Bitcoin. Backwards-compatibility is desirable in decentralized networks because it means that a majority of the network
can upgrade to the new protocol and non-upgraded nodes can still verify blocks and execute transactions. Below we describe a scheme that, when added to Nakomoto consensus, assures Assumption~\ref{prop:adversary-advantage}. The proof of the following theorem is in the full version of our paper.

\begin{theorem}\label{thm:property-vdf}
Assumption~\ref{prop:adversary-advantage} is satisfied
by our VDF-based scheme.
\end{theorem}

%For this protocol, we use the alternative definition of 
%a synchronous network which is that the network delay is 
%at most $\Delta$. Thus, a round consists of $\Delta$ timesteps.
%\nn{Note that the following protocol could be added as a \emph{soft fork} to Nakamoto consensus; it is backwards compatible. A majority of the network could update to using this symmetry breaking strategy without requiring the entire network to upgrade at once.} 
% Oh that's cool! I added this note in because it seems compelling.
%\qq{De-emphasize cryto tools or write it in a distributed
%friendly manner.}
%We set the difficulty of our VDF $\func$ (defined as in
%Definition~\ref{def:vdfs-short}) to be $\Delta$
%where $\Delta$ is the maximum delay per round. 
\myparagraph{VDF-Scheme Overview}
The VDF-scheme works intuitively as follows. 
We number the rounds beginning with round $0$. 
All nodes have the genesis block $B_0$ in their local
chains in round $0$ and starting mining blocks in round
$1$. In round $0$, the VDF output is computed using $0$
as the input.
During each round $j > 0$,
each node computes a 
VDF output, $y_j$, (using $\func$) 
for the current round $j$ where the input to $\func$
is the output
of the VDF, $y_{j - 1}$, 
from the previous round concatenated with
the round number, $j$. Both inputs are necessary;
the output of the VDF from the previous round 
ensures that we cannot compute the VDF output for this 
round until we have obtained the output for the previous
round, and the round number is necessary to ensure that 
the output is \emph{not} used for a future round. Once
the VDF output is computed, each honest node attempts
to mine a block using the VDF output as part of the 
input to the mining attempt. This also ensures
that the block generation rate of honest
nodes is upper bounded by $np$.
Then, each node which successfully mines a block
sends the new chain to all other nodes. 

All honest nodes
verify that each chain satisfies two conditions: 
\begin{enumerate}
    \item Let $o_1, \dots, o_\ell$ be the VDF outputs
    contained in blocks $B_1, \dots, B_\ell$, respectively, of a chain $C$ (the genesis block does not contain a VDF output). Let $r_1, \dots, r_\ell$ be the rounds where $o_1, \dots,
    o_{\ell}$ were computed, respectively. Then, 
    $r_1 < \dots < r_{\ell - 1} < r_{\ell}$.
    \item $o_i$ is the VDF output computed 
    from round $r_i \geq i - 1$.
\end{enumerate}
The honest nodes also check all proofs included in the chains, confirming
that the VDF outputs are correctly computed and the blocks are correctly mined using the VDF outputs. 
An honest node discards any chain which does not pass
verification.

\vspace{1em}
\myparagraph{Pseudocode}
The precise pseudocode of our VDF-based scheme is given
below.
Using $\func$, each honest 
node $i$ performs the following:

\begin{enumerate}
    \item Initially, all honest nodes use 
    input $0$ at the start of the protocol to obtain
    %use their VRF to 
    %obtain a uniformly at random\footnote{This is not
    %exactly precise, but one can think of it that way.}
    output $y_0 = \func(0)$ for round $0$.
    \item Let $d_j = \func(y_{j-1})$ be the output of the VDF for round $j$ and $y_j = d_{j - 1}|j$.\footnote{Here, $a|b$ is the commonly used notation indicating
    concatenation between $a$ and $b$.} $i$ stores $y_{j}$.
    \item When $i$ mines a block $B_j$, $i$ includes 
    the output $y_{j-1} = d_{j - 1} | j$ from the previous round
    in $B_j$, ie.\ $B_j$
    is mined with $y_{j-1}$ as part of the input.
    \item Each node which successfully mines a block adds
    the mined block to its local chain. Then, it broadcasts
    its local chain to all other nodes.
    \item For each longest chain received, each node verifies
    the following:
    \begin{enumerate}
        \item Let $o_1, \dots, o_\ell$ be the VDF outputs
        stored in each block
        in order starting with the first block and 
        ending with the $\ell$-th block. Let $r_1, \dots, r_\ell$
        be the rounds associated with the VDF output.
        Then, $r_\ell > r_{\ell - 1} > \cdots > r_1$. \label{step:greater}
        \item The $k$-th block in the chain (starting from
        the genesis block) is mined using $y_{k'}$ from round
        $k' \geq k - 1$. \label{step:correct-vdf-mined}
        %\nn{You are re-using $i$ here to index rounds. Earlier it referred to the node. Use a different var or don't bother naming the node.}
        \item The proofs of the VDF output 
        and the mining output
        are correct, i.e.\ the block is 
        correctly mined using the 
        corresponding VDF output.
    \end{enumerate}
    \item If $i$ receives a chain where 
    more than one block in the chain is mined with the 
    same $y_j$ (for any $j$ smaller than the current round), the node discards the chain. 
    %\nn{wait -- there is no pki and no identities for nodes. How do I know "who" sent me a block??} \qq{I rephrased it so it's more clear, you're checking no two blocks in the entire chain are mined
    %using the same VDF, so you wouldn't need identities in this case.}
    \item At the end of round $j$, $i$ sets $y_{j + 1} \leftarrow \func(y_{j}) | j + 1$ and
    %where $j$ is the round number and
    begins computing the next value $\func(y_{j+1})$
    %\nn{should be $j+1$}
    using $y_{j+1}$ 
    %\nn{$j$} 
    as input.  
    %stores $\func(y_{j +1})$ 
    %\nn{$j+1$}
    %as the VDF output for round $j + 1$. 
    %\nn{$j+1$}
\end{enumerate}

Due to space constraints, we do not include the proof of~\cref{thm:property-vdf}; 
please find the full proofs in the full version of our paper. However, 
the intuition for our proof is straightforward. Items~\ref{step:greater}
and~\ref{step:correct-vdf-mined} ensure that no chain accepted by an
honest node contains more than one block per VDF output. Setting
the difficulty of the VDF to the duration of the round ensures that
at most one VDF output is produced during a round. Together, these
two observations 
prove~\cref{thm:property-vdf},
namely, that any chain held by an honest node can 
be extended by at most one block each round.

\section{Discussion}
\label{sec:discussion}

\myparagraph{Validation and Communication Costs}
A higher $p$ means a faster block rate and thus more blocks. 
The validation and bandwidth complexity of Nakamoto protocols are proportional to block size and the number of blocks that are mined, since each miner validates and then communicates every mined block to all other miners (in practice, nodes do not necessarily gossip shorter chains, and taking advantage of nodes' memory overlap can help reduce block transfer size~\cite{compactblocks}). One needs to determine the optimal value of $p$ that trades off validation and bandwidth complexity and chain growth. This work expands the space of $p$ to consider.

\medskip

\myparagraph{Other Symmetry-Breaking Strategies}
Here we consider three other symmetry-breaking strategies with high $p$.
\emph{First-seen} is where
all honest nodes take the first chain out of the longest-length chains they see, and \emph{lexicographically-first} is where honest nodes take the lexicographically-first chain of the set of longest chains according to some predetermined ordering, for example alphabetically.
%and add a block to that chain or that symmetry is broken arbitrarily. 
%Honest nodes can have some predetermined system of ordering chains lexicographically. For
%example, such a system could define a chain to be lexicographically first
%if its most recent block appears alphabetically before the most recent blocks
%of the other chains, ensuring a unique lexicographical order of the chains.
Intuitively, the adversary can control the network and thus cause different honest nodes to see different chains of the same length first for first-seen, impacting common-prefix, or grind on blocks to always produce the lowest lexicographically-ordered chain for lexicographically-first, impacting chain-quality.
A third strategy is to use a \emph{global-random-coin}: Suppose that all nodes have access to a permutation oracle $\perm$ that returns a permutation sampled uniformly at random of a number of elements passed into it \emph{where any subset of elements obey the same partial ordering.} With $\perm$ symmetry-breaking is trivial since all honest nodes will agree on the result of the coin flip.
Furthermore, if the coin is fair, then the number
of honest blocks added to the chain is proportional
to the fraction of honest nodes. However, in reality, it is difficult and oftentimes infeasible to ensure such a strong guarantee.

\myparagraph{Conclusion} In this work we show that unlike previously thought, convergence opportunities are not necessary to make chain progress. We use \emph{coalescing random walks} to analyze the correctness of Nakamoto consensus under a regime of puzzle difficulty previously thought to be untenable, expanding the space of $p$ for protocol designers.

\bibliographystyle{abbrv}
\bibliography{SM}

\appendix
\begin{center}
\bf \Large Appendices      
\end{center}

\iflong
\section{Additional Definitions}
\label{app: additional definitions}

\subsection{VDFs}
The formal definition of VDFs is presented below.

\begin{definition}[Verifiable Delay Functions~\cite{vdf}]\label{def:vdfs}
    A VDF $V = \left(\setup, \eval, \veri\right)$ is a triple of algorithms that perform the following:
    \begin{enumerate}
        \item $\setup(\secp, \dif) \rightarrow \textbf{pp} = (\ek{}, \vk{})$: The $\setup$ algorithm takes
            as input a security parameter $\secp$ and a desired difficulty level $\dif$ and produces
            public parameters consisting of an evaluation key $\ek{}$ and a verification key $\vk{}$.
            $\setup$ is polynomial time with respect to $\secp$ and $\dif$ is subexponentially-sized
            in terms of $\secp$. The public parameters specify an input space $\sX$ and an output space $\sY$.
            $\sX$ is efficiently sampleable. If secret randomness is used in $\setup$, a trusted setup might
            be necessary.
        \item $\eval(\ek{}, x) \rightarrow (y, \pi)$: $\eval$ takes an input $x \in \mathcal{X}$ (in the
            sample space of inputs) and the evaluation key and produces an output $y \in \mathcal{Y}$ (in the
            sample space of outputs) and a (possibly empty) proof $\pi$. $\eval$ may use random
            bits to generate $\pi$ but not to compute $y$. $\eval$ runs in parallel time $\dif$
            even when given $poly(\log(\dif), \secp)$ processors for all $\textbf{pp}$ generated
            by $\setup(\secp, \dif)$ and $x \in \sX$. \label{item:difficulty}
            \item $\veri(\vk{}, x, y, \pi) \rightarrow \{Yes, No\}$: $\veri$ is a deterministic algorithm
                that takes the verification key $\vk{}$, an input $x$, the output $y$, and proof $\pi$ and outputs
                $Yes$ or $No$ depending on whether $y$ was correctly computed from via $\eval$. $\eval$ runs in
                time $O(\log(\dif))$. 
    \end{enumerate}
    Furthermore, $V$ must satisfy the following properties:
    \begin{enumerate}
        \item \textbf{Correctness} A VDF $V$ is correct if for all $\secp, \dif$,
            parameters $(\ek{}, \vk{}) \xleftarrow{R} \setup(\secp, \dif)$,
            and all $x \in \sX$, if $(y, \pi) \xleftarrow{R} \eval(\ek{}, x)$,
            then $\veri(\vk{}, x, y, \pi) \rightarrow Yes$.
        \item \textbf{Soundness} A VDF is sound if for all algorithms $\alg$ that run in time
            $O\left(poly(\dif, \secp)\right)$
            \begin{align*}
                \Prob\left[
                    \begin{array}{l}
                        \veri(\vk{}, x, y, \pi) = Yes \\
                        y' \neq y
                    \end{array}
                \middle\vert
                    \begin{array}{l}
                        \mathbf{pp} = (\ek{}, \vk{}) \xleftarrow{R} \setup(\secp, \dif)\\
                        (x, y', \pi') \xleftarrow{R} \alg\left(\secp, \mathbf{pp}, \dif\right)\\
                        (y, \pi) \xleftarrow{R} \eval(\ek{}, x)
                    \end{array}
                \right] \leq negl(\secp).
            \end{align*}
        \item \textbf{Sequentiality} A VDF is $(p, \sigma)$-sequential if no adversary $\alg = (\alg_0, \alg_1)$
            with a pair of randomized algorithms $\alg_0$,
            which runs in total time $O(poly(\dif, \secp))$, and $\alg_1$, which runs in parallel time $\sigma(t)$
            on at most $p(t)$ processors, can win the following game with probability greater than $negl(\secp)$:
            \begin{flalign*}
                &\textbf{pp} \xleftarrow{R} \setup(\secp, \dif) \\
                &L \xleftarrow{R} \alg_0(\secp, \textbf{pp}, \dif) \\
                &x \xleftarrow{R} \mathcal{X} \\
                &y_{\alg} \xleftarrow{R} \alg_1(L, \textbf{pp}, x).
            \end{flalign*}
            $\alg = (\alg_0, \alg_1)$ wins the game if $(y, \pi) \xleftarrow{R} \eval(\ek{}, x)$ and $y_{\alg} = y$.
    \end{enumerate}
\end{definition}

There are many implementations in the literature of VDFs (e.g.\ \cite{P18,W20}). We do not provide these implementations here 
as it is out-of-scope for our paper, but please refer
to these papers for contructions of VDFs that satisfy
the above properties.

\subsection{Tail Bounds}
We use the following variant of Hoeffding's inequality.  % for Bernoulli random variables which is standard in the literature:

\begin{theorem}[Hoeffding's inequality]\label{def:hoeffding}
Let $Y_1, \dots, Y_n$ be $n$ independent, identically
distributed random variables drawn from a Bernoulli distribution with parameter $p$, $Y_i \stackrel{\text{i.i.d.}}{\sim} \text{Ber}(p)$. If $S_n = \sum_{i = 1}^n Y_i$, then $S_n \sim 
\text{Binom(n, p)}$, $\Expect(S_n) = np$, and 
\begin{align*}
    \prob{\frac{1}{n}\left|\sum_{i = 1}^n Y_i - pn \right|\geq \eps} \leq 2e^{-2n\eps^2}.
\end{align*}
\end{theorem}
\iffalse
We use the following variant of the Chernoff bound
which is standard in the literature:

\ls{Hoeffiding's inequality is enough for our case. For Bernoulli random variables, Hoeffding's inequality is better. }
\begin{theorem}[Chernoff Bound]\label{def:chernoff}
    Let $Y_1, \dots, Y_m$ be $m$ independent random variables
    that take on values in $[0, 1]$ where $\mathbb{E}[Y_i] = p_i$
    and $\sum_{i = 1}^m p_i = P$. For any $0 < \gamma \leq 1$,
    the multiplicative \emph{Chernoff bound} gives
    \begin{align*}
        \Prob\left[\sum_{i=1}^m Y_i > (1+ \gamma)P\right] < \exp\left(-\gamma^2 P/3\right)
    \end{align*}

    and
    \begin{align*}
        \Prob\left[\sum_{i=1}^m Y_i < (1-\gamma)P\right] < \exp\left(-\gamma^2 P/2\right).
    \end{align*}
\end{theorem}
\fi

\subsection{The Bitcoin Blockchain System}\label{sec:bitcoin-overview}
In this section, for completeness, we provide a high-level overview of the Bitcoin Blockchain
System. The below is mainly to serve as a reminder of the Bitcoin 
protocol for those unfamiliar with it. 
% In our paper, we will mostly
% focus on a simplified version of the protocol (defined 
% in Definition~\ref{def:bitcoin-system}) 
% and spend most of the 
% paper discussing the symmetry breaking strategies.

\myparagraph{High-Level Description}
\label{subsec: high-level description}

The \emph{nodes} in the system represent miners in the Bitcoin cryptosystem 
who mine blocks filled with requests from \emph{clients}. Clients
represent payers who would like to fulfill some \emph{transactions}. The
client issues a \emph{write-request} whenever it wants to send a transaction
to a miner. The miner then attempts to \emph{mine} a block containing
the value of the transaction. Specifically, the following set of steps occur:

\begin{itemize}
\item The payer submits a write-request to the system with a valid transaction as the write ``value'' they want to add to the public ledger. 
\item Every honest miner $i$:
\begin{itemize}
\item has a $\mempool$ which contains a collection of multi-cast transactions received by this miner. Notably, due to issues such as network failures and messages delay, the $\mempool$ kept by different miners might not be identical, and
\item keeps a local valid blockchain $\calC_i$. 
\end{itemize}
\item In each round, each of the miners: 
\begin{enumerate}
\item Blockify its local $\mempool$ (i.e., creates a block of appropriate size that contains a sub-set of the transactions in $\mempool$) and removes those blockified transactions from $\mempool$.
\item Try to add this new block to its local chain $\calC_i$. 
\item If the miner successfully extends its local chain, it multi-casts the updated chain to other miners. 
\item Wait to receive multi-casted chains from others and update its local chain to be the chain that is the longest among the received chains and its current local chain. If there are multiple longest chains, use a {\bf symmetry breaking} mechanism to choose one of them as its new local chain.   
\end{enumerate}
\end{itemize}
In the Bitcoin system, oftentimes, the symmetry is broken in an arbitrary manner, i.e., if there is a tie, an honest node chooses an arbitrary longest chain (e.g.\ the chain it received first). In an adversarial setting, 
this symmetry-breaking strategy could potentially lead to honest nodes
choosing \emph{different} chains frequently. It turns out that this 
symmetry-breaking
rule, with high probability, can guarantee safety as long as it is sufficiently hard to successfully mine a block. However, this is not the
case when the probability of successfully mining a block is large. In 
fact, for such instances, it is important to consider specific 
symmetry-breaking strategies and how they affect the system.

\section{Honest Majority Assumption}\label{sec: honest maj}
The honest majority assumption in the seminal \cite{garay2015bitcoin} is presented below for completeness. 
For ease of comparison, we use the same notation as that in \cite{garay2015bitcoin}
Let $f_0$ be the probability at least one honest node succeeds in finding a proof-of-work (pow) in a round. 
In \cite{garay2015bitcoin}, the notion of the {\em advantage of honest participants} is used, denoted by $\delta$. It is used to bound $\frac{b}{n-b}$. In particular, $\delta$ is chosen so that $\frac{b}{n-b} \le 1-\delta$ always holds. 
\begin{assumption}[Honest Majority Assumption \cite{garay2015bitcoin}]
\label{ass: honest majority}
Given an $\epsilon\in (0,1)$, $n$, and $p$, the maximal number of corrupted nodes $b$ satisfies: 
\begin{itemize}
\item $3f_0+3\epsilon < \delta \le 1$, and 
\item $b \le (1 - \delta)(n-b)$. 
\end{itemize}
\end{assumption}
Notably, by definition of $\delta$, the second bullet in Assumption \ref{ass: honest majority} always holds. Hence, for fixed $\epsilon\in (0,1)$, $n$, and $p$, the real constraint on $b$ is the relation assumed in the first bullet of Assumption \ref{ass: honest majority}. 
%Clearly, exactly $1/2$ corrupt nodes cannot be tolerated. As in that case $\delta=0$ contradicts the required condition in Assumption \ref{ass: honest majority} wherein 
% \[
% 3f+3\epsilon<\delta =0. 
% \]
%A natural question is: Up to what percentage, the system can tolerate corrupt nodes? 

\begin{proof}[\bf Proof of Proposition \ref{prop: honest majory necessity}] 
By Assumption \ref{ass: honest majority}, it holds that 
\begin{align}
\label{eq: honest majority}
3f+3\epsilon<\delta \le 1.    
\end{align}
As $\epsilon>0$, \eqref{eq: honest majority} implies that $3f<\delta$. 
Let $f$ denote the probability at least one honest node succeeds in finding a pow in a round.
We have $f\ge 1- \pth{1-p}^{(n-b)}$. So 
\begin{align*}
 3\pth{1- \pth{1-p}^{(n-b)}} \le 3f<\delta. 
\end{align*}
Equivalently, 
\begin{align*}
\log (1-\delta/3) < (n-b)\log(1-p),  
\end{align*}
for arbitrary base of log as long as the base is $\ge 1$. By Taylor expansion, we have 
\[
(n-b)p < \frac{\delta/3}{1-\delta/3} \le \frac{\delta}{2},  
\]
where the last inequality follows from the fact that $\delta\in (0,1)$. 

\end{proof}

\noindent{\em On Remark \ref{rmk: existing slow}: } 
% Proposition \ref{prop: honest majory necessity} implies that for Bitcoin backbone to work with the probably guarantee assured in \cite{garay2015bitcoin}, unless $\frac{b}{n}$ is {\em non-trivially} bounded above from $\frac{1}{2}$, the mining rate would be extremely low -- even much lower than the commonly believed $\Theta(\frac{1}{n})$. 
% 
To see the claim in Remark \ref{rmk: existing slow}, consider the boundary case where $n=2b+1$ -- the honest nodes barely make it to be the majority of the system. In this case, the upper bound of $p$ in Proposition \ref{prop: honest majory necessity} is 
\begin{align*}
\frac{n-2b}{2(n-b)^2} = \frac{1}{2(n- \frac{n-1}{2})^2} = \frac{2}{(n+1)^2}.      
\end{align*}
Thus, %in a round, the honest nodes can at most collectively mine $\frac{1}{n+1}$ blocks in expectation. That is, 
in expectation, it takes at least $n+1$ rounds for the honest nodes to mine a block collectively. 
Such a low block generating speed makes it unlikely to have multiple longest chains unless the network delay is very serious. This observation also justifies why the choice of symmetry breaking rules does not matter much in \cite{garay2015bitcoin,pass2017analysis}. 

This observation holds not only for the boundary case when $n=2b+1$ but also for more general $b$. For ease of illustration, let's consider the sequence of $b_k$ for $k=1, \cdots, \lfloor\frac{n-1}{2}\rfloor$ with $b_k:= \lfloor\frac{k}{2k+1}n \rfloor$. Without loss of generality, assume that $\frac{k}{2k+1}n$ is an integer for all $k$ under consideration. For a system with up to $b_k$ corrupted nodes, the upper bound in Proposition \ref{prop: honest majory necessity} lies in between $\pth{\frac{1}{2(k+1)n}, ~ ~ \frac{1}{(k+1)n}}$. 
%It is shown in \cite{garay2015bitcoin} that as long as the majority of the nodes are honest, desired properties such as liveness and persistence are preserved; see \cite[Lemmas 24 and 25]{garay2015bitcoin} for details. However, they did not characterize how $\delta$ -- the competing advantage of the honest nodes -- affects these properties. 
In a sense, the Honest Majority Assumption (formally stated in \ref{ass: honest majority}) requires the mining puzzle becomes harder as $b \to \frac{1}{2}n$. That is, Assumption \ref{ass: honest majority} requires the system to trade off liveness for tolerating more corrupt nodes. 
\fi

\section{Proof of Theorem \ref{lm: p=1 inde random}}
\label{app: proof of p =1}

\begin{proof}[\bf Proof of Theorem \ref{lm: p=1 inde random}]
We formalize the arguments of % and fill in the details of 
the main proof ideas in Section \ref{sec: warmup}.  
Let $\{c_1, \cdots, c_n\}$ be a set of $n$ different colors. We associate each node in the system with a color. We use $(c_i, t)$ to denote the block generated by honest node $i$ during round $t$ and $\pth{c_0, 0}$ to denote the genesis block. 
We use $(c_i, t)\to (c_{i^{\prime}}, t-1)$ to denote the event that block $(c_i, t)$ is attached to block $(c_{i^{\prime}}, t-1)$, which occurs with probability $\frac{1}{n}$ under our symmetry-breaking rule. 
To quantify the maximal inconsistency of the longest chains of round $T$, we consider the following auxiliary random process. It  can be easily shown that there is a bijection between the sample paths of the Bitcoin blockchain protocol and the sample paths of this auxiliary process, and that the auxiliary process and 
the original blockchain protocol with random symmetry breaking have the same probability distribution. 

\vskip 0.2\baselineskip
% \paragraph{Auxiliary random procedure $\text{No.}1$} 
\noindent {\em \underline{Auxiliary random procedure}:}  For any given $T\ge 1$, do the following:  \\
% \begin{itemize}
%     \item 
     (i) Let each color generate a block for each of the rounds in $\sth{1, 2, \cdots, T}$; \\ 
    %\item 
    (ii) Attach each of the block $\pth{c_i, 1}$ for $i=1, \cdots, n$ to the genesis block $\pth{c_0, 0}$; \\  
    (iii) For each $t\ge 2$ and each $\pth{c_i, t}$, attach it to one of the blocks $\sth{\pth{c_i, t-1}, i=1, \cdots, n}$ uniformly at random  (i.e., with probability $1/n$).  
%\end{itemize}
% 
% 

% 
% 
\vskip 0.2\baselineskip
\noindent {\em \underline{Connecting to coalescing random walks}:} Here, we formally quantify the connection between the maximal inconsistency among the longest chains of round $T$ with the coalescing time of $n$ random walks on an $n$-complete graph. 
Since $p=1$ and there is no adversary, the number of longest chains received by each honest node at each round is $n$. 
Let $C(T, c_1), \cdots, C(T, c_n)$ be the $n$ longest chains of round $T$ ending with blocks $\pth{c_1, T}, \cdots, \pth{c_n, T}$, respectively. 
We first show that each of these $n$ chains can be coupled with a random walk on the $n$-complete graph. 
Without loss of generality, let's consider $C(T, c_1)$ which can be expanded as 
\begin{align}
\label{eq: chain expansion}
\mbox{$\squeezespaces{0.3}
C(T, c_1) : = \pth{c_0, 0} \gets \pth{c_{i_1}, 1} \gets \cdots \gets \pth{c_{i_{t-1}}, t-1} \gets \pth{c_{i_{t}}, t} \gets \cdots \gets   \pth{c_{i_{T-1}}, T-1} \gets \pth{c_1, T},$}  
\end{align}
where $c_{t}$ is the color of the $\pth{t+1}$-th block in the chain. Note that the chain $C(T, c_1)$ is random because the sequence of block colors 
$c_0c_{i_1}\cdots c_{i_{t-1}} c_{i_t} \cdots c_{i_{T-1}}c_1$ is random. Moreover, the randomness in $C(T, c_1)$ is fully captured in the randomness of the block colors. We have 

\begin{align*}
& \mbox{$\squeezespaces{0.4}\prob{C(T, c_1)  =\pth{c_0, 0} \gets \pth{c_{i_1}, 1} \gets \cdots \gets \pth{c_{i_{t-1}}, t-1} \gets \pth{c_{i_{t}}, t} \gets \cdots \gets   \pth{c_{i_{T-1}}, T-1} \gets \pth{c_1, T}}$} \\
& \overset{(a)}{=} \prob{\pth{c_0, 0} \gets \pth{c_{i_1}, 1}} \prod_{t=2}^{T}\prob{\pth{c_{i_{t-1}}, t-1} \gets \pth{c_{i_{t}}, t}}  \\
& = \prod_{t=2}^{T}\prob{\pth{c_{i_{t-1}}, t-1} \gets \pth{c_{i_{t}}, t}},
\end{align*}

where the last equality is true as $\prob{\pth{c_0, 0} \gets \pth{c_{i_1}, 1}} =1$, and the equality (a) holds because
under our symmetry-breaking rule,
%, in each round $t$, each node chooses which of the longest chains received at the beginning of round $t$ to extend on uniformly at random. Thus, 
neither the previous history up to round $t$ nor the future block attachment choices after round $t$ affects the choice of the chain extension in round $t$. 
Moreover, the probability of any realization of the color sequence $c_0c_{i_1}\cdots c_{i_{t-1}} c_{i_t} \cdots c_{i_{T-1}}c_1$ (i.e., a sample path on the block colors in Bitcoin) is $\pth{\frac{1}{n}}^{T-1}$.
Let's consider the complete graph with vertex set $\sth{c_1, c_2, \cdots, c_n}$. 
Under our symmetry breaking rule, the backwards color sequence $c_1c_{i_{T-1}}\cdots c_{i_t}c_{i_{t-1}} \cdots c_{i_1}$ (without considering the genesis block) is a random walk on the $n$-complete graph starting at vertex $c_1$. Similarly, we can argue that $C(T, c_2), \cdots, C(T, c_n)$ correspond to $n-1$ random walks on the $n$-complete graphs starting at vertices $c_2, \cdots, c_n$, respectively. As argued in the {\bf main proof ideas} paragraph, these $n$ random walks are not fully independent. In fact, they are coalescing random walks, and their coalescence is exactly the maximal inconsistency among the longest chains of round $T$. 

With the above connection of 
%The above arguments connect 
the longest chain protocol augmented by uniformly-at-random symmetry breaking with coalescing random walks.
We conclude by applying Theorem \ref{thm: coalsecing rw complete}. 
%We apply Theorem \ref{thm: coalsecing rw complete} to conclude Theorem \ref{lm: p=1 inde random}.
\end{proof}

% \section{p=1: Adversary-prone}
% \label{subsec: p=1 and Adversary-prone}

\iflong
\section{Missing proofs and auxiliary results for Section \ref{subsec: general p and Adversary-free}}
\label{app subsec: general p and Adversary-free}
% 
% 
% Similar to the proof of the special case when $p=1$, in our proof of general $p\in (0, 1)$, we will couple the growth the of the common prefix in Bitcoin protocols with the coalescing time random walks on complete graphs. The major differences from the proof of $p=1$ are: (1) instead of the standard coalescing random walks, we need to work with a lazy version of it, formally defined in Definition \ref{def: lazy-coalescing-random-walk}; (2) there is no fixed correspondence between a color and a node -- in our proof of general $p$, the correspondence is round-specific rather than fixed throughout the entire dynamics; (3) there is no bijection between a sample path of the Bitcoin dynamics and that of the backwards coalescing random walks, thus, we need to rely on stochastic dominance to build up the connection of these two dynamcis. Fortunately, though the proof for $p\in (0, 1)$ is significantly more involved, the high-level ideas in the proof of the special case $p=1$ applies.  

% 

% 
% 

\begin{proof}[Proof of Lemma \ref{lm: lazy-coalescence-time}]
To characterize the coalescence time, similar to the analysis in \cite{cooper2013coalescing}, for any given $k\in \{ 1, \cdots, n_g\}$, we construct a larger graph $Q=Q_k=(V_Q, E_Q)$, where $V_Q=V^k$ and two vertices $\bm{v}, \bm{w}\in V^k$ if $\sth{v_1, w_1}, \cdots, \sth{v_k, w_k}$ are edges of $G$. Let $M_k$ be the time until the first meeting in the original graph $G$. 
Let $S\subseteq V_Q$ denote the set of all possible configurations of the locations of the $n_g$ random walks at the first meeting, 
\begin{align}
\label{def: Qk-first-meeting}
S_k = \sth{(v_1, \cdots, v_k): v_i = v_j ~~~ \text{for some } 1\le i < j \le k}. 
\end{align}
It is easy to see that there is a direct equivalence between the $u$-lazy random walks on $G$ and the single $u$-lazy random walk on $Q$. 
Since $Q$ is a complete graph with self-loops, the limiting distribution of lazy random walk on $Q$ is the same as the standard random walk on $Q$. 
Let $\bm{\pi}^{Q} \in \reals^{|V^k|}$ be the stationary distribution of a standard random walk on $Q$ and let $\pi^Q_{S_k} = \sum_{\bm{v} \in S_k} \pi^Q_{\bm{v}}$. 
By \cite[Lemma 4]{cooper2013coalescing}, we know that for any $1\le k \le k^*$ where $k^* \triangleq \max\{2, \log n_g\}$, it holds that 
\begin{align*}
\pi^Q_{S_k}  \ge \frac{k^2}{8n_g}.     
\end{align*}
Let $H_{\bm{v}, S_k}$ denote the hitting time of vertex set $S_k$ starting from vertex $v$ and let 
$$H_{\bm{\pi}}^Q(H_{S_k}) = \sum_{\bm{v}\in V^k} \bm{\pi}^Q_{\bm{v}} H_{\bm{v}, S_k}$$ 
denote the expected hitting time of $S_k$ from the stationary distribution $\bm{\pi}^Q$. 
From \cite[Lemma 2.1]{aldous2002reversible} and the fact we can contract the vertex set $S_k$ into one pseudo vertex, similar to \cite[proof of Theorem 2]{cooper2013coalescing}, 
we have that 

\begin{align*}
\mathbb{E}_{\bm{\pi}^Q}[H_{S_k}]  &= \frac{\sum_{t=0}^{\infty} \pth{P_{S_k}^t(S_k) -\bm{\pi}^Q_{S_t}}}{\bm{\pi}^Q_{S_k}}  = \frac{\sum_{t=0}^{\infty} \pth{(1-u)^t + \pth{1 - (1-u)^t}\bm{\pi}^Q_{S_k} - \bm{\pi}^Q_{S_k} }}{\bm{\pi}^Q_{S_k}} \\
& \le \frac{8n_g}{k^2}  \frac{1}{u}\pth{1- \bm{\pi}^Q_{S_k}} \le \frac{8n_g}{uk^2}. 
\end{align*}
In addition, by conditioning on whether the particles stay at their initial locations or not, we have 
\begin{align*}
\expect{M_k} = \pth{1-u}\pth{1 + \expect{M_k}} + u \pth{1+\mathbb{E}_{\bm{\pi}^Q}\pth{H_{S_k}}},      
\end{align*}
which implies that 
\begin{align*}
\expect{M_k}  \le \frac{1}{u}\pth{1+ \frac{8n_g}{k^2}} = O\pth{\frac{n_g}{uk^2}}. 
\end{align*}
Thus, for any $k$ such that $1\le k \le k^* = \sth{2, \log n_g}$,  we have 
\begin{align*}
\expect{C_k} \le \sum_{s=2}^k \expect{M_s} \le O(n_g/u).     
\end{align*}

Let $\calW_u$ be a lazy random walk on the complete graph $G$ with initial location $u$.  In each round, with probability $(1-u)$, $\calW_u$ stays at its current location and with probability $u$ it  moves to one of the current neighbors (including self-loops) uniformly at random. Let $\bm{\pi}^G$ the limiting distribution of the location vertex of $\calW_u$. 
By \cite[Eq.(8)]{cooper2013coalescing}, its mixing time is $t_{mix} = \frac{3\log n_g}{\log(1/(1-u))}$, i.e., for any given $u\in V$, when $t\ge \lceil  \frac{3\log n_g}{\log(1/(1-u))} \rceil$, 
\begin{align*}
\|P_u^t - \bm{\pi}^G\|_1 &= \sum_{v\in V}\abth{P_u^t(v) - \bm{\pi}^G_v} \\
& = \abth{1-\pi^G_u}(1-u)^t + \sum_{v: v\in V, v\not=u}\abth{\pth{1-(1-u)^t}\bm{\pi}^G_v - \bm{\pi}^G_v} \\
& \le 2(1-u)^t 
\le \frac{2}{n_g^3} \le \frac{1}{n_g^2}.   
\end{align*} 
Here, with a little abuse of notation, we use $P_u^t$ to denote the distribution of the state of $\calW_u$ at round $t$. 
Let $t^* = k^* \log n_g \pth{k^* t_{mix} + 3 \mathbb{E}_{\bm{\pi}^Q}\pth{H_{S_{k^*}}}}$. Following the arguments in \cite[Section 5]{cooper2013coalescing}, we have 
\begin{align*}
C(n_g) &\le 4t^* + \expect{C_{k^*}} \\
& \le 4 \log n_g \pth{k^* t_{mix} + 3 \mathbb{E}_{\bm{\pi}^Q}\pth{H_{S_{k^*}}}} + O(n_g/u) \\
& \le \frac{4\log^4n_g}{\log \frac{1}{1-u}} + 12 \log^2n_g \frac{8n_g}{u\log^2 n_g} + O(n_g/u)\\
& \le \frac{4\log^4n_g}{u} + \frac{96n_g}{u} + O(n_g/u) \\
& = O(n_g/u),
\end{align*}
where the last inequality follows from $\log {1/(1-u)} \ge u$.

\end{proof}

The following lemma will be used in the proof of Theorem \ref{lm: general p inde random}

\begin{lemma}
\label{lm: balls and bins}
Suppose that there $k$ balls. 
Let $X$ be the number of non-empty bins if we throw each of the $k$ balls into $b$ bins, where $k\le b$, uniformly at random. 
Let $\tilde{X}$ be the number of non-empty bins if we throw each of the $k$ balls into $b+\Delta$ bins, where $\Delta \in \naturals$, uniformly at random. 
Then $\tilde{X}$ first-order stochastically dominates $X$. 
\begin{align*}
\prob{X \le l} \ge \prob{\tilde{X} \le l} ~~ \forall l.  
\end{align*}

\end{lemma}
\begin{proof} 

Intuitively speaking, since $b <b+\Delta$, collisions are more likely to occur when fewer bins are available. Hence, $X_2$ first-order stochastically dominates $X_1$. For the sake of peace of mind, a formal proof is given below. 

Let's consider the mental process wherein we throw the balls into bins one by one. 
Let $Y_t$ after we throw $t$ balls into $b$ bins. Similarly, $Z_t$ be the number of nonempty bins we throw $t$ balls into $b+\Delta$ bins.  
We show Lemma \ref{lm: balls and bins} by induction on $t$. \\
Clearly, $Y_1 = 1= Z_1$. \\ 
Induction hypothesis: Suppose for $t\le k-1$, there exists a coupling between the marginal probabilities of the above two ball throwing processes such that under this coupling 
\begin{align}
\label{eq: induction hypothesis}
Y_t \le Z_t. 
\end{align}

When $Y_t \le Z_t-1$, by Eq.\eqref{eq: induction hypothesis} and the monotonicity of $Y$ and $Z$, it holds that 
$Y_{t+1} \le Y_t +1 \le Z_t \le Z_{t+1}$. 
It remains to consider the case where $Y_t = Z_t.$ It is easy to see that $Y_{t+1}=Y_t$ if the $t+1$--th was thrown into the existing non-empty bins, which occurs with probability $\prob{Y_{t+1}=Y_t} = \frac{Y_t}{b}$. 
Similarly, $\prob{Z_{t+1}=Z_t} = \frac{Z_t}{b+\Delta}.$ For ease of exposition, let $Z_{t+1}=Z_t = \gamma$. 
Consider the following coupling: 
\begin{itemize}
    \item If the $(t+1)$--th ball of the second bins-and-balls process is thrown into the $Z_t$ existing nonempty bins, then put the $(t+1)$--th ball of the first bins-and-balls process uniformly at random into its $Y_t$ existing nonempty bins. 
    \item If the $(t+1)$--th ball of the second bins-and-balls process is thrown into an empty bin, then with probability $\frac{\Delta \gamma}{b(b+\Delta-\gamma)}$ put the $(t+1)$--th ball of the first bins-and-balls process uniformly at random into one existing nonempty bin. With probability $1-\frac{\Delta \gamma}{b(b+\Delta-\gamma)}$, put the $(t+1)$--th ball of the first bins-and-balls process
    into one empty bin uniformly at random. 
\end{itemize}
It is easy to see that in the above coupling, the $(t+1)$--th ball of the first bins-and-balls process is throw into a bin (regardless whether it is empty or not) with probability $\frac{1}{b}$. 
Moreover, with this coupling and the induction hypothesis, we know that 
\begin{align*}
Y_{t+1} \le Z_{t+1},    
\end{align*}
completing the induction proof. Hence, $X=Y_k \le Z_k=\tilde{X}$. Therefore, 
\begin{align*}
\prob{\tilde{X}\le l} \le \prob{X\le l}, \, \, \forall \, l,  
\end{align*}
i.e., $\tilde{X}$ first-order stochastically dominates $X$. 

\end{proof}
\else
%\section{Proof of~\cref{lm: general p inde random}}
\fi

\begin{proof}[\bf Proof of Theorem \ref{lm: general p inde random}]
For any $t$, the expected length of a longest chain is 
$1 + \pth{1-\pth{1-p}^n}t.$ 
% \begin{align*}
% 1 + \pth{1-\pth{1-p}^n}t.     
% \end{align*}
% 
When $p< \frac{4\ln 2}{n}$, we can use Poisson approximation to approximate the distribution of number of blocks in each round. A straightforward calculation shows that the probability of having exactly one block in a round is $np \exp\pth{-np}. $
% \begin{align*}
% np \exp\pth{-np}.    
% \end{align*}
Thus, in expectation, the maximal inconsistency is at most $\frac{1}{np \exp\pth{-np}}$. 
Henceforth, we restrict our attention to the setting where $p\ge \frac{4\ln 2}{n}$ and quantify the expected maximal inconsistency among the longest chains of round $t$.  
We first consider a coarse analysis whose arguments are similar to the proof of Theorem \ref{lm: p=1 inde random} and derive a bound on the maximal inconsistency via stochastic dominance. 
Though the obtained bound could be very loose, based on the insights obtained in this coarse analysis, we can come up with a much fine-grained analysis, which significantly improves the bound on maximal inconsistency. % -- bringing the bound down to $O(np)$. 

\vskip 0.3\baselineskip 
\noindent{\bf A coarse analysis:} 
Let $\{c_1, \cdots, c_n\}$ be a set of $n$ different colors. We temporarily associate each node in the system with a color. \footnote{In our fine-grained analysis, the color of a block will be re-assigned. } 
If node $i$ mines a block during round $t$, we denote this block by $(c_i, t)$. In addition, we use $\pth{c_0, 0}$ to denote the genesis block. 
We use $(c_i, t)\to (c_{i^{\prime}}, t-1)$ to denote the event that both blocks $(c_i, t)$ and  $(c_{i^{\prime}}, t-1)$ exist and that block $(c_i, t)$ is attached to block $(c_{i^{\prime}}, t-1)$, which, under our symmetry-breaking rule,  occurs with probability 
\begin{align*}
\frac{\indc{\text{node $c_i$ mines a block during round $t$}}\indc{\text{node $c_{i^{\prime}}$ mines a block during round $t-1$}}}{\sum_{i^{\prime}=1}^n \indc{\text{node $c_{i^{\prime}}$ mines a block during round $t-1$}}}.    
\end{align*}
Notably, in the Bitcoin protocol, there are two sources of randomness: (1) the randomness in generating blocks and (2) the randomness in the block attachments. 
To quantify the maximal inconsistency of the longest chains of round $T$, we consider the following auxiliary random process. It can be easily shown that there is a bijection between the sample paths of the Bitcoin blockchain protocol and the sample paths of this auxiliary process, and that the auxiliary process and 
the original blockchain protocol with random symmetry breaking have the same probability distribution. 

\vskip 0.2\baselineskip
% \paragraph{Auxiliary random procedure $\text{No.}1$} 
\noindent {\em \underline{Auxiliary random procedure}:}  For any given $T\ge 1$, do the following:  \\
% \begin{itemize}
%     \item 
     (i) For each of the rounds in $\sth{1, 2, \cdots, T}$, let each node/color generate a block with probability $p$ independently of other nodes and independently across rounds. For ease of exposition, we refer to the blocks mined in round $t$ as the blocks in layer $t$.  
     \\ 
    %\item 
    (ii) Attach each of the block $\pth{c_i, 1}$, if exists, for $i=1, \cdots, n$ to the genesis block $\pth{c_0, 0}$; \\  
    (iii) For each $t\ge 2$ and each $\pth{c_i, t}$ that exits, attach it to one of the blocks in layer $(t-1)$. If layer $(t-1)$ is empty, let 
    \begin{align*}
        t^{\prime} \triangleq \max\sth{r:  \text{layer $r$ is nonempty } \text{and} ~ r\le t}, 
    \end{align*}
    and let each existing $\pth{c_i, t}$ uniformly at random chooses one ancestor block in block layer $t^{\prime}$. 

\vskip 0.2\baselineskip
\noindent {\em \underline{Connecting to coalescing random walks}:} We first build a coarse connection between the maximal inconsistency among the longest chains of round $T$ with the coalescing time of $n$ random walks on an $n$-complete graph. A much fine-grained connection to coalescing random walks on $2np$-complete graph in given in fine-grained analysis part of this proof. 
It is easy to see that the number of blocks mined in each round $t$, denoted by $N_t$, follows the $\Binom(n,p)$ distribution. 
Without loss of generality, we assume that $N_T\not=0$. If this does not hold, then we can replace $T$ by the most recent round $T^{\prime}$ such that $N_{T^{\prime}}\not=0$ and the remaining proof goes through. 
Since there is no adversary, the number of longest chains at the end of round $T$ is $N_T$, each of which ends with a block in block layer $T$. 
Let $C(T, c_{1}^{\prime}), \cdots, C(T, c_{N_T}^{\prime})$ be the $N_T$ longest chains of round $T$ ending with blocks $\pth{c_{1}^{\prime}, T}, \cdots, \pth{c_{N_T}^{\prime}, T}$, respectively. 
We first show that each of these $N_T$ chains can be coupled with a process that is a variant of a random walk on the $n$-complete graph. 
Without loss of generality, let's consider $C(T, c_1^{\prime})$ which can be expanded as 
\begin{align}
\label{eq: chain expansion general p}
C(T, c_1^{\prime}) : = \pth{c_0, 0} \gets \pth{c_{i_1}, 1} \gets \cdots \gets \pth{c_{i_{k-1}}, k-1} \gets \pth{c_{i_{k}}, k} \gets \cdots \gets   \pth{c_{i_{K-1}}, K-1} \gets \pth{c_{1}^{\prime}, T},
\end{align}
where $c_{k}^{\prime}$ is the color of the $\pth{k+1}$-th block in the chain and $K$ is the number of non-empty block layers under event $E$ -- in the realization of block mining. Recall that for general $p\in (0, 1)$ there are two sources of randomness (1) the randomness in block generating and (2) the randomness in block attachment. 
Consequently, the sequence of block colors $c_0c_{i_1}\cdots c_{i_{k-1}} c_{i_k} \cdots c_{i_{K-1}}c_1^{\prime}$ is random in that, roughly speaking, the ``feasibility'' of $c_{i_k}$ is determined by whether node $i_k$ mines a block during round $k$ or not, and the ordering of the ``feasible" colors is determined by the attachment choices. 
Let $E$ be any realization of the block mining for the first $T$ rounds, which corresponds to any realization of step (i) of the auxiliary process.
We have 
\begin{align*}
& \prob{C(T, c_1^{\prime}) \mid E} \\ 
& = \prob{\pth{c_0, 0} \gets \pth{c_{i_1}, 1} \gets \cdots \gets \pth{c_{i_{k-1}}, k-1} \gets \pth{c_{i_{k}}, k} \gets \cdots \gets   \pth{c_{i_{K-1}}, K-1} \gets \pth{c_1^{\prime}, T} ~ \mid ~ E} \\
& \overset{(a)}{=} \prob{\pth{c_0, 0} \gets \pth{c_{i_1}, 1} \mid E} \pth{\prod_{k=2}^{K-1}\prob{\pth{c_{i_{k-1}}, k-1} \gets \pth{c_{i_{k}}, k} \mid E}} \prob{\pth{c_{i_{K-1}}, K-1} \gets \pth{c_1^{\prime}, T} ~ \mid ~ E }  \\
& = \pth{\prod_{k=2}^{K-1}\prob{\pth{c_{i_{k-1}}, k-1} \gets \pth{c_{i_{k}}, k} \mid E}} \prob{\pth{c_{i_{K-1}}, K-1} \gets \pth{c_1^{\prime}, T} \mid  E }, 
\end{align*}
where the last equality is true as $\prob{\pth{c_0, 0} \gets \pth{c_{i_1}, 1} \mid E} =1$, and the equality (a) holds because
under our symmetry-breaking rule,
%, in each round $t$, each node chooses which of the longest chains received at the beginning of round $t$ to extend on uniformly at random. Thus, 
neither the previous history up to round $t$ nor the future block attachment choices after round $t$ affects the choice of the chain extension in round $t$. 
Moreover, the conditional probability of any realization of the color sequence $c_0c_{i_1}\cdots c_{i_{k-1}} c_{i_k} \cdots c_{i_{K-1}}c_1^{\prime}$ conditioning on $E$ (i.e., a sample path on the block colors in Bitcoin) is $\prod_{t: 2\le t \le T, \& n_t\not=0} \frac{1}{n_t}$ where $N_t = n_t$ for all $t\in \{2, \cdots, T\}$. 
Let's consider the complete graph with vertex set $\sth{c_1, c_2, \cdots, c_n}$. 
Under our symmetry breaking rule, conditioning on $E$, the backwards color sequence $c_1^{\prime}c_{i_{K-1}}\cdots c_{i_k}c_{i_{k-1}} \cdots c_{i_1}$ (without considering the genesis block) is 
a walk, though not the standard random walk, of length $T$ on the $n$-complete graph with initial location $c_1^{\prime}$. 
Similarly, we can argue that $C(T, c_2^{\prime}), \cdots, C(T, c_{N_T}^{\prime})$ correspond to $\pth{N_T-1}$ walks on the $n$-complete graphs starting at vertices $c_2^{\prime}, \cdots, c_{N_T}^{\prime}$, respectively. Similar to the argument in the proof of Theorem \ref{lm: p=1 inde random}, conditioning on $E$, there is an one-to-one correspondence between the event of the forking of the chains $C(T, c_1^{\prime}), \cdots, C(T, c_{N_T}^{\prime})$ and the event of coalescence of the backwards walks; that is, the maximal inconsistency of the longest chains of round $T$ is the same as the coalescence time of the $N_T$ walks on the $n$-complete graphs. 
These $N_T$ walks are more likely to coalesce than the standard random walks whose transition probability is $1/n$, whereas under any $E$, $N_t = n_t \le n$ for $t=1, \cdots, T$; this fact can be formally shown via 
\iflong
Lemma \ref{lm: balls and bins}. 
\else
Lemma 30 of \fl.
\fi
%This is because for any $E$, $N_t = n_t \le n$ for $t=1, \cdots, T$ and each block is attached to one of the block in the preceding block layer uniformly at random. 
Hence, 
the conditioning on $E$, the maximal inconsistency is upper bounded by the coalescence time of the standard random walks on $n$-complete graph with $n$ particles, which is $O(n)$.
Since this is true for all possible block mining realization $E$, we conclude that the maximal inconsistency is upper bounded by $O(n)$.

\vskip 0.3\baselineskip 
\noindent{\bf A fine-grained analysis:} 
Let $E$ any realization of the block mining for the first $T$ rounds. 
If not explicitly mentioned, the following arguments are stated conditioning on $E$. 
To conclude the proof, towards the end of this proof, we take average over all possible events $E$. 

%As mentioned in the end of the {\em coarse analysis}, we can bound the expected coalescing time of the $N_T$ backwards walks by the coalescing time of $n$ random walks on a $n$-complete graph. 
The bound on the maximal inconsistency is $O(n)$ which could be loose for a wide value range of $p$.  
% However, the obtained bound turns out to be loose. 
This is because the upper bound on $N_t$ is loose. 
Recall that $N_t$ is a $\Binom(n,p)$. Thus, $\expect{N_t} = np \ll n$ as long as $p=o(1)$. 
% for sufficiently large $T$, with high probability, $N_t\approx np$ for many rounds in $\{1, \cdots, T\}$. 
%Moreover, we care about whether these walks coalesce or not rather than at what specific vertex they coalesce. 
Observing this, in this {\em fine-grained analysis}, we first construct a lazy version of the $N_T$ backwards walks whose expected coalescence time is at least the coalescence time of the original $N_T$ backwards walks. Then re-color the mined blocks so that the re-colored lazy version of the $N_T$ original walks are walks on at most $2np$ colors only. 
Then, we connect these lazy walks with the a lazy version of random walks each of which, if not stay at their current locations concurrently, moves to one of the neighboring colors (including its current color) with probability $\frac{1}{2np}$.   
Finally, by changing of the order of taking expectation, we show that the maximal inconsistency is upper bounded by the expected coalescence time of the $(1-2 \exp\pth{-\frac{1}{3}np})$-lazy random walks. 
We conclude the proof of this theorem by applying Lemma \ref{lm: lazy-coalescence-time}. 
%and couple the coalescence time of the $N_T$ lazy backwards walks with the coalescence time of the lazy coalescing random walks on the $2np$-complete graph.  
%Here, WLOG, in the following argument we assume that $2np$ is an integer. The analysis can be easily adapted to general $np$. 

% We first couple these $n_T$ backwards walks $\tilde{c}_i \tilde{c}_{i_{K-1}} \cdots \tilde{c}_{i_k}\tilde{c}_{i_{k-1}} \cdots \tilde{c}_{i_1}\tilde{c_1}$ with a lazy version of these backwards walks as follows: \\

\vskip 0.2\baselineskip
\noindent{\em \underline{Lazy walks construction:} } 
Consider the following $n_T$ lazy version of the backwards coalescing walks. 
%$2np$-complete graph with initial locations the arbitrary $\max\{2np, n_T\}$ on the graph. 
For each $s=1, \cdots, K-1$, if $n_{K-s} \le 2np$, each of the remaining walks moves to one of the block color in layer $(K-s)$ uniformly at random. If two or more walks visit the same color, then these walks coalesce into one. 
If $n_{K-s} > 2np$, we let the remaining walks stay at their current color vertices. 
% we do not attach any of the blocks in layer $K-s+1$ to blocks in layer $(K-s)$. 
Clearly, this lazy version of the $N_T$ backwards walks are more likely to coalesce than the original $N_T$ walks. Let $C_{l}(E)$ denote the expected number of backwards steps until all the lazy $N_T$ walks coalesce.

\vskip 0.2\baselineskip
\noindent{\em \underline{Color re-assignment:} } 
Next we show that, under any $E$, $C_{l}(E)$ is upper bounded by the expected coalescence time of $\max\{2np, n_T\}$ lazy random walks on the $2np$-complete graph. Towards this, we first do color-reassignment, detailed as follows.  
Let $\{\tilde{c}_1, \cdots, \tilde{c}_n\}$ be a set of $n$ different colors.  
% WLOG, we assume there exists a universe of nodes $U$ and the set of current nodes in the system is a unknown subset of this universe set $U$. Index the nodes in the universe from $m_1$ to $m_{|U|}$ regardless whether a node is a node or not. 
We (re-)assign a color to each of the mined block in different block layers as follows:  
%$\sth{c_1, \cdots, c_n}$ to the mined blocks as follows. 
Assign color $\tilde{c}_1$ to the genesis block. For each $t=1, \cdots, T$ such that $N_t\not=0$, let $i_1^t < \cdots <i_{N_t}^t$ be the indices of the nodes/original colors each of which successfully mines a block during round $t$.   
% \[
% i_1(t) \triangleq  \min\sth{i\in [n]: ~ \text{node $m_i$ successfully mines a block at round $t$}},
% \]
% where $[n] \triangleq \{1, \cdots, n\}$. 
% 
Re-assign colors $\tilde{c}_1, \cdots, \tilde{c}_{N_t}$ to blocks $\pth{c_{i_1^t}, t}, \cdots, \pth{c_{i_{N_t}^t}, t}$. 
The re-colored blocks are denoted as $\pth{\tilde{c}_1, t}, \cdots, \pth{\tilde{c}_{N_t}, t}$, respectively. 

Notably, different from the color assignment we used in the proof for the case when $p=1$, under the above color re-assignment rule, the blocks mined by the same node at different rounds could be assigned different colors. 
Fortunately, it is easy to see that the blocks attachments are independent of the color assignments. In particular, it is still true that the maximal inconsistency among the longest chains of round $T$ is the same as the coalescence time of the corresponding backwards walks on those colors. Moreover, it is still true that the expected coalescence time of the re-colored $N_T$ walks is upper bounded by $C_l(E)$ the lazy version of the re-colored $N_T$ walks. More importantly, the re-colored lazy $N_T$ walks, expect for their initial colors, are the walks on at most $2np$ colors in each round. 

Consider the following $\max\sth{2np, N_T}$ lazy coalescing random walks on the $2np$-complete graph with arbitrary but distinct initial locations. % of the $\max\sth{2np, N_T}$ walks. 
For each $s=1, ... T$, if the above $N_T$ lazy version of the walks on colors stay at their own locations concurrently (i.e., $N_{T-s}\ge 2np$) or $N_{T-s} =0$, then each of the remaining random walks on the $2np$-complete graph also stay at their current locations. If otherwise (i.e., each of the lazy walks on colors moves to one of the colors assigned to the blocks in the proceeding non-empty layer uniformly at random), we let each of the remaining walks on the $2np$-complete graph moves to one of the $2np$ vertices uniformly at random.
\iflong
By Lemma \ref{lm: balls and bins},
\else
By Lemma 30 of \fl,
\fi
we know the expected coalescence time of the $N_T$ lazy walks on colors is upper bounded by that of the $\max\sth{2np, N_T}$ walks on the $2np$-complete graph which is again upper bounded by the expected coalescence time of $2np$ lazy walks on the $2np$-complete graph, denoted by $C_{l,np}(E)$. 

Next we consider averaging over all the realizations of $E$. For any given $T$, with the above arguments, we know that  the expected maximal inconsistency is upper bounded by 
\begin{align*}
\sum_{E}C_{l,np}(E)\prob{E}.     
\end{align*}
%where $E$ is one realization of the block mining for the first $T$ rounds. 
Note that by construction, $C_{l,np}(E)$ depends on $E$ only through the number of blocks mined in each round.  In particular, it only depends on each $N_t$ in whether $1\le N_t<2np$ holds or not. So we have 
\begin{align*}
&\sum_{E}C_{l,np}(E)\prob{E}  = \sum_{n_1, \cdots, n_{T-1}}C_{l,np}(n_1, \cdots, n_{T-1})\prob{N_t=n_t, \forall ~ t\le T-1}\\
& = \sum_{n_1, \cdots, n_{T-2}}\pth{\sum_{n_{T-1}} C_{l,np}(n_1, \cdots, n_{T-1}) \prob{N_{T-1} = n_{T-1}}}   \prob{N_t=n_t, \forall ~ t\le T-2} \\
& = \sum_{n_1, \cdots, n_{T-2}} \left( C_{l,np}(n_1, \cdots, n_{T-1}\in \{2np, \cdots, n\}\cup \{0\}) \prob{N_{T-1} \ge  2np  \text{ or } N_{T-1} =0}    \right. \\
&  +\left.    C_{l,np}(n_1, \cdots, n_{T-1}\in \{1, \cdots, 2np-1\}\ ) \prob{1\le N_{T-1} <  2np }     \right)   \prob{N_t=n_t, \forall ~ t\le T-2}.
\end{align*}
Note that when $n_{T-1}\in \{2np, \cdots, n\}\cup \{0\}$, each of the $2np$ random walks stay at their initial locations concurrently with occurs with probability $\prob{N_{T-1} \ge  2np  \text{ or } N_{T-1} =0}$, and when $n_{T-1}\in \{1, \cdots, 2np-1\}$, each of the $2np$ random walks take one step standard coalescing random walks. That is, the $2np$ random walks are performing one step $\pth{1 - \prob{N_{T-1} \ge  2np  \text{ or } N_{T-1} =0}}$-lazy random walk on the $2np$-complete graph. Since $N_t$ is $\iid$ across $t$, we can repeat this argument for $T-1$ times. In fact, we can exchange the order taking expectation over $E$ and taking expectation over the realization of the walks. Hence, $\sum_{E}C_{l,np}(E)\prob{E}$ equals the expected coalescence time of $2np$ $\pth{1 - \prob{N_{T-1} \ge  2np  \text{ or } N_{T-1} =0}}$-lazy random walks. By Lemma \ref{lm: lazy-coalescence-time}, we know that 
\begin{align*}
    \sum_{E}C_{l,np}(E)\prob{E} = O\pth{\frac{2np}{\pth{1 - \prob{N_{T-1} \ge  2np  \text{ or } N_{T-1} =0}}}}. 
\end{align*}

In addition, we have 
\begin{align*}
\prob{N_t\ge 2np, \, \text{or}\, N_t=0} & =  \prob{N_t\ge 2np} + \prob{N_t=0} \\
& = \prob{N_t\ge 2np} + \pth{1-p}^n  = \prob{N_t\ge 2np} + \exp\pth{-n\log\frac{1}{1-p}} \\
& \le \exp\pth{-\frac{1}{3}np} + \exp\pth{-n\log\frac{1}{1-p}} \le 2 \exp\pth{-\frac{1}{3}np},
\end{align*}
where the last inequality holds because $p\le \log \frac{1}{1-p}$ when $p\in [0, 1)$. 
So, it holds that 
\begin{align*}
    \sum_{E}C_{l,np}(E)\prob{E} = O\pth{\frac{2np}{\pth{1 - 2 \exp\pth{-\frac{1}{3}np}}}},  
\end{align*}
proving the theorem.

\end{proof}

\iflong
\section{Missing proofs and auxiliary results for Section \ref{subsec: general p and Adversary-prone}}
\label{app: general p and Adversary-prone}
% 
% 

% \noindent {\em Selective relay rule: }
% Let us make the following algorithmic changes: At each honest node, say node $j$, for each iteration $t\ge 1$: If the length of its received longest chains at the beginning of iteration $t$ is longer than its local chain, then, before making any block mining attempt in iteration $t$, the node $j$ multi-casts one of the received longest chains to others. 
% % 

% \begin{definition}[Adversary advantages]
% \label{def: adversarial ad}
% Given the number of blocks generated in each round, we construct the following random process which we refer to as the {\em Adversary advantages} sequence. 
% Let $\sth{\calN(t)}_{t=0}^{\infty}$ be the random process defined as 
% \begin{itemize}
% \item $\calN(0)=0$. 
% \item For $t\ge 1$, 
% \begin{align*}
% \calN(t) =   
% \begin{cases}
% \calN(t-1) + 1, & ~ ~ \text{if in round $t$ only the corrupt nodes successfully mine a block} \\
% \calN(t-1), & ~ ~ \text{if in round $t$ either no nodes mine a block or both the corrupt }  \\
% &  ~ ~ \text{nodes and the honest nodes mine blocks} \\
% \max\{\calN(t-1) -1, ~ 0\}, & ~ ~  \text{if in round $t$ only the honest nodes successfully mine a block}. 
% \end{cases}
% \end{align*}
% \end{itemize}
% \end{definition}

\begin{proof}[\bf Proof of Lemma \ref{lm: new proof}]
For each round $\tau$, %\qq{for each round $\tau \in [t_0, t_k]$?},  
the following holds: 
\begin{itemize}
\item If no blocks are mined, then the lengths of the adversarial longest chains and the honest longest chains (longest chains kept by an honest node) are not changed. 
\item If both the honest and the corrupt nodes mine a block, then the lengths of the adversarial longest chains increases by 1, and length of the honest longest chains (longest chains kept by an honest node) increases by {\em at least} 1. 
%\qq{Why is it ``at least''? By the same logic that adversaries cannot add more than one block to a chain using the VDF, an honest node also cannot add more than one block to a chain.} 
To see the later, let's denote the length of the longest chains at the honest nodes at round $(\tau-1)$ by $\ell_{(\tau-1)}$ and the length of the longest chains at the honest nodes at round $\tau$ by $\ell_{\tau}$. By the selective relay rule 
%\qq{What is the ``selective relay rule''? Is it the rule that all honest nodes broadcast their longest chain? If so we should defined \emph{selective relay rule} as such.} 
specified right before Definition \ref{def: adversarial ad}, by the beginning of round $\tau$, every honest node has received a chain that is at least $\ell_{(\tau-1)}$. 
%\qq{What does this notation mean? $\tau$ is defined as a round number?}. 
If the adversary does not release a prefix of an adversarial chain of length $>\ell_{(\tau-1)}$, the length of the longest chains kept by the honest nodes at round $\tau$ is $\ell_{\tau} = \ell_{(\tau-1)}+1$. Otherwise, due to the longest chain policy, it holds that $\ell_{\tau}>\ell_{(\tau-1)}+1$. 
\item If only corrupt nodes mine a block, then the adversary can grow the length of the adversarial longest chains by 1. 
%\qq{This is an assumption that adversaries can only grow the adversarial change by $1$. We need either a construction or an assumption for this to hold.} \qq{Pull over the VDF definition.} 
The length of the longest chains at the honest nodes is unchanged. 
\item If only honest nodes mine a block, then the length of the honest longest chains increase by at least 1. The formal argument follows the same as the proof of the later part of the second bullet. The length of the longest chains at the corrupt nodes (adversary) is unchanged. 
\end{itemize}
Let $t^{\prime}\triangleq \max\{t^{\prime}: \calN(t^{\prime})=0~ \text{and } t^{\prime}\le t\}$. 
Let $t^{\prime} = t_0  < t_1 < \cdots < t_k$ be the round indices of the jumps 
of the random process ${\calN(\tau)}_{\tau=0}^{\infty}$, i.e., 
\begin{align*}
\calN(t_0)\not=\calN(t_0 + 1), \cdots, \calN(t_k-1)\not=\calN(t_k).      
\end{align*}
By the construction of ${\calN(\tau)}_{\tau=0}^{\infty}$, we know that from round $t^{\prime}+1$ to round $t$, the number of rounds in which only corrupt nodes mine a block is $\calN(t)$ larger than the number of rounds in which only honest nodes mine a block. %\qq{Why is it by construction of ${\calN(\tau)}_{\tau=0}^{\infty}$? ${\calN(\tau)}_{\tau=0}^{\infty}$ is a random process which 

Therefore, we know that as long as at round $t^{\prime}$, the length of the adversarial longest chain is no longer than the length of the honest longest chains, we can conclude that at the end of round $t$, the length of the adversarial longest chains is at most $\calN(t)$ blocks longer than the length of the honest longest chains. 
%\qq{Again, this is assuming that adversaries can only grow their chains by at most $1$ each round. We need to clearly state this as an assumption or show a construction that ensures this.}  
% 
It remains to show the following is true: 
%\noindent{\bf Claim:} 
{\em ``At round $t^{\prime}$, the length of the adversarial longest chain is no longer than the length of the honest longest chain. ''}

Let $\calN(t^{\prime})$ be the $k$--th time starting from round 0 such that 
\begin{align*}
\calN(t') =0 ~ ~ \text{and} ~ ~ \calN(t'+1) =1. 
\end{align*}
%\qq{What is $\tau$ here? Is it supposed to be $\calN(t') = 0$ and $\calN(t'+ 1) = 1$?}
If $k=1$, by the arguments in the above four bullets, we know  the claim holds. Let's assume this claim holds for general $r$. We next prove it holds for $r+1$. Let $t^{\prime\prime}$ be the $r$--th time starting from round 0 such that
\begin{align*}
\calN(t^{\prime\prime}) =0 ~ ~ \text{and} ~ ~ \calN(t^{\prime\prime}+1) =1. 
\end{align*}
By induction hypothesis, we know that at round $t^{\prime\prime}$, the length of the adversarial longest chain is no longer than the length of the honest longest chain. By the first part of the proof of Lemma \ref{lm: new proof}, we know at round $t^{\prime}$, the length of the adversarial longest chain is no longer than the length of the honest longest chain. Thus, the proof of the claim is complete. 
\end{proof}

The following lemma follows from Hoeffding's inequality. 
\begin{lemma}
\label{lm: quantify J} 
With probability at least $\pth{1- \exp\pth{-\frac{(p^*)^2M}{2}}}$, it holds that 
\begin{align}
\sum_{i=1}^M \indc{\calJ(m) \not =\calJ(m-1)} \ge \frac{1}{2}p^*M.      
\end{align}
\end{lemma}
% \begin{proof}
% Since $\indc{\calJ(m) \not =\calJ(m-1)}$'s are $\iid$ Bernoulli random variables, it holds that 
% \begin{align*}
% \expect{\sum_{m=1}^M \indc{\calJ(m) \not =\calJ(m-1)}} = M \expect{\indc{\calJ(1) \not =\calJ(0)}}  = Mp^*. 
% \end{align*}
% We conclude this lemma by applying Hoeffding's inequality. 
% \end{proof}

% 
% 
Next we prove Lemma \ref{lm: the number of coalescing opportunities}. 
\begin{proof}[\bf Proof of Lemma \ref{lm: the number of coalescing opportunities}]
From \cite[Chapter 4.10]{hajek2015random} we know that $\sth{\calJ(t)}_{t=0}^{\infty}$ has a corresponding jump process (also referred to as the embedded chain) that describes, conditioning on state changes, how $\calJ(t)$ jumps among different states. We also know that this jump process is a simple random walk. 
% 
% \begin{definition}[Jump process of $\sth{\calJ(t)}_{t=0}^{\infty}$]
Concretely, if $\calJ(t-1)\not=\calJ(t)$, we say a jump occurs at $t$. 
Let $T_r$ denote the number of rounds that elapses between the $r$-th and $r+1$-th jumps of $\calJ$.
Let $\tilde{\calJ}(r)$ denote the state after $r$ jumps. 
By definition, $\sth{\tilde{\calJ}(r)}_{r=0}^{\infty}$ is the jump process of $\sth{\calJ(t)}_{t=0}^{\infty}$. 
It is easy to see that $T_r$ is a geometric random variable with parameter $p^*$. 
Also, $T_1, T_2, \cdots$ are $\iid$ distributed. 
By \cite[Proposition 4.9]{hajek2015random}, we know 
\begin{align*}
\tilde{\calJ}(r) = 
\begin{cases}
0, ~ ~ &\text{if }r=0; \\
\tilde{\calJ}(r-1)+1, ~~ &\text{with probability } p_{+1}/p^*; \\
\tilde{\calJ}(r-1)-1, ~ ~ &\text{with probability } p_{-1}/p^*.  \\
\end{cases}
\end{align*}
Let $\delta_r  = \tilde{\calJ}(r) - \tilde{\calJ}(r-1)$ for any $r\ge 1$. It is easy to see that $\delta_r$ is a Bernoulli random variable supported on $\{-1, +1\}$ and $\prob{\delta_r=1}= p_{+1}/p^*$.  
For a given $K\ge 0$ jumps, by Hoeffding's inequality, we know that with probability at least $\pth{1-\exp\pth{-\frac{K(p_{+1}/p^*-p_{-1}/p^*)^2}{8}}}$, 
the following is true   
\begin{align*}
\sum_{r=1}^K \delta_r \ge  \frac{(p_{+1}/p^*-p_{-1}/p^*)}{2} K.     
\end{align*}
Setting $K = \frac{1}{2}p^*M$, we have 
\[
\sum_{r=1}^K \delta_r  =  \sum_{r=1}^{\frac{1}{2}p^*M}\pth{\tilde{\calJ}(r) - \tilde{\calJ}(r-1)} = \tilde{\calJ}(\frac{1}{2}p^*M) - \tilde{\calJ}(0) = \tilde{\calJ}\pth{\frac{1}{2}p^*M}. 
\]
In addition, from Lemma \ref{lm: quantify J}, we know that with probability at least $(1-\exp\pth{-((p^*)^2M)/2})$, it holds that 
\begin{align}
\sum_{i=1}^M \indc{\calJ(m) \not =\calJ(m-1)} \ge \frac{1}{2}p^*M.      
\end{align}
Thus, $\calJ(M) \ge \tilde{\calJ}\pth{\frac{1}{2}p^*M}$, proving the lemma.

% By union bound, we conclude that with probability at least $1 - \exp\pth{-(p^*M)/2} - exp\pth{-\frac{K(p_{+1}^*-p_{-1}^*)^2}{8}}$, 
% $\calJ(M) \ge \frac{(p_{+1}^*-p_{-1}^*)p^*}{4} M$, 
% proving Lemma \ref{lm: the number of coalescing opportunities}. 
\end{proof}

\begin{lemma}
\label{prop: zero advantage}
For any $t> 1$ round  such that $\calN(t) = 0$ and at least one block is mined, let $\text{NB}(t)$ be the number of blocks mined by the honest nodes in round $t$ and $\text{AB}(t)$ be the number of blocks mined by the corrupt nodes in round $t$. 
For $\tilde{t} \triangleq \min\{t^{\prime}:\,\, t^{\prime} \le t-1\}$ such that at least one block is mined, let $\text{AB}(\tilde{t})$ be the number of blocks mined by corrupt nodes in round $\tilde{t}$.
Then it is true that the number of longest chains at the end of round $t$ is at most $\pth{\text{NB}(t) +  \text{AB}(t)+\text{AB}(\tilde{t})}$. In particular, if $\calN(t-1)=0$, then the number of longest chains at the end of round $t$ is at most $\pth{\text{NB}(t) + \text{AB}(t)} $. Moreover, all of these longest chains end with blocks generated in round $t$. 
\end{lemma}
\begin{proof}%[Proof of Lemma \ref{prop: zero advantage}]
One of the following two cases hold.  \\
\noindent{Case 1:} Suppose that $\calN(t-1) = 0$.  
By Lemma \ref{lm: new proof}, at the beginning of round $t$, the longest chains received (including their own local chains) by each of the honest node are of the same length. Thus, the number of longest chains at the end of round $t$ is $\pth{\text{NB}(t)+\text{AB}(t)}$ -- the number of blocks generated during round $t$. 

\noindent{Case 2:} Suppose that $\calN(t-1) \not= 0$.
Recall that  $\tilde{t}\triangleq \max\{t^{\prime}: t^{\prime}\le t-1\}$ such that at least one block is generated during round $\tilde{t}$. By Definition \ref{def: adversarial ad} and the fact that $\calN(t)=0$, we know that $\calN(\tilde{t}) = 1$. From Lemma \ref{lm: new proof}, we know that, at the end of round $\tilde{t}$, the length of the adversarial longest chains of round $\tilde{t}$ is at most one block longer than the local chains at the honest nodes. The adversary can choose either to hide these longest adversarial chains to some/all honest parities or to release those chains to all of the honest nodes. The number of such adversarial longest chains is at most $\text{AB}(\tilde{t})$ -- the number of blocks mined by the corrupt nodes in round $\tilde{t}$. Hence the number of longest chains 
%received by an honest node (including its local chain) 
at the beginning of round $t$ is at most $\text{AN}(\tilde{t}) + \text{AB}(t) + \text{NB}(t)$. 
\end{proof}

\begin{lemma}
\label{lm: common longest chains}
Let $i_1$ and $i_2$, where $i_1\not=i_2$, be two arbitrary honest nodes. For any $(t-1)\ge 1$ such that at least one block is mined during round $t-1$, let $\calC_1(t)$ and $\calC_2(t)$ be the sets of longest chains received by honest nodes $i_1$ and $i_2$, respectively, at round $t$ before mining (including the chains sent by others at the end of round $t-1$ and forwarded at the beginning of round $t$). If $\calN(t-1) = 0$, it holds that 
\begin{align}
\label{eq: overlap longest chains}
\abth{\calC_1(t)\cup \calC_2(t)} \ge \text{NB}(t-1), 
\end{align}
and 
\begin{align}
\label{eq: upper bounds}
\max\{\abth{\calC_1(t)},\, \abth{\calC_2(t)}\}  ~ \le ~ \text{NB}(t-1) + \text{AB}(t-1) + \text{AB}(\widetilde{t-1}),  
\end{align}
where $\widetilde{t-1}: \max\{t^{\prime}: ~ t^{\prime} \le t-2\}$. 
Moreover, if $\calN(t-2) = 0$, then 
\begin{align}
\label{eq: upper bounds ideal}
\max\{\abth{\calC_1(t)},\, \abth{\calC_2(t)}\}  ~ \le ~ \text{NB}(t-1) + \text{AB}(t-1). 
\end{align}
\end{lemma}
\begin{proof}%[Proof of Lemma \ref{lm: common longest chains}]
Since $\calN(t-1)=0$, %we know that each of the honest node receives at least one longest chain up to the beginning of round $t$. That is, 
both $\calC_1(t)$ and $\calC_2(t)$ are subsets of the longest chains at the end of round $(t-1)$. 
In addition, from Lemma \ref{prop: zero advantage}, we know that the total number of longest chains at the end of round $t-1$ in the system is at most $\text{NB}(t-1) + \text{AB}(t-1) + \text{AB}(\widetilde{t-1})$, proving Eq.\eqref{eq: upper bounds}. As each of the honest nodes who successfully mines a block during round $(t-1)$ multi-casts its local chain to others, both $i_1$ and $i_2$ will receive all of the longest chains that end with an honest block. That is, 
\begin{align*}
\abth{\calC_1(t)\cup \calC_2(t)} \ge \text{NB}(t-1),     
\end{align*}
proving Eq.\eqref{eq: overlap longest chains}.
Since $\calN(t-1)=0$, it is easy to see that $\text{NB}(t-1)\ge 1$. 

When $\calN(t-2)=0$, it holds that each of the honest node receives at least one longest chain up to the beginning of round $t-1$. Thus, the total number longest chains (including both the ones extended by the honest nodes and the ones extended by the corrupt nodes) is $\text{AB}(t-1) + \text{NB}(t-1)$. Since a corrupt node can choose to not to multi-cast its local chain, it holds that 
\begin{align*}
\max\{\abth{\calC_1(t)},\, \abth{\calC_2(t)}\}  ~ \le ~ \text{NB}(t-1) + \text{AB}(t-1).     
\end{align*}

Note that as a corrupt node can arbitrarily choose, independently of others, which subset of honest nodes to send its local chain, the honest nodes $i_1$ and $i_2$ can be two different subsets of the longest chains at the end of round $t-1$. 

\end{proof}

\section{VDF-Scheme Proofs}
The above VDF-based scheme presented in Section~\ref{sec:vdf-protocol}
gives us 
the following properties which will be
crucial in obtaining our bounds for the case when adversaries are 
present. We present the proofs omitted from 
Section~\ref{sec:vdf-protocol} in this section.

\begin{lemma}\label{lem:vdf-seq}
For every round $t \geq 0$,
the local chain of every honest node can contain at most
one block per VDF output, with all but negligible probability
in $\secp$.
\end{lemma}

\begin{proof}
\cref{step:greater} ensures that the VDF outputs stored in the blocks
are computed in strictly increasing order. This means that no 
chain verified by an honest node can contain
two blocks which contain the same VDF output. Finally, with all but
negligible probability in $\secp$ (by Definition~\ref{def:vdfs-short}), 
no two VDF outputs will be 
equal. The lemma follows. 
\end{proof}

\begin{lemma}\label{lem:adversary-two-blocks}
The adversary can add at most one block to the same chain (that
will be verified by honest nodes)
during any round $r \geq 0$, with all but negligible probability
in $\secp$. In other words,
let $t$ be the current round, any adversarial
chain verified by honest nodes cannot have length greater than
$t$.
\end{lemma}

\begin{proof}
By Lemma~\ref{lem:vdf-seq}, no adversary can add 
two blocks with the same VDF output to the same chain. 
Thus, in any adversarial chain verified by honest nodes, 
the adversary can add at most $1$ block per round, and the 
chain will have length at most $t$ (where $t$ is the 
current round). 
\end{proof}

\begin{lemma}\label{lem:future}
Let $t$ be the current round. No node can add a block 
containing a VDF output, $o_{t'}$, 
from a round $t' \geq t$. 
\end{lemma}

\begin{proof}
We prove this via induction. In the base case when the chain only 
contains the genesis block, the adversary has access only to the genesis
block and no other blocks. In this case, the adversary could not have 
mined any other blocks because it did not have time to obtain the VDF output
for the next block, with all but negligible probability in $\secp$.
Let our induction hypothesis be that the 
adversary cannot use $d_{t}$ to mine a block in round $t$.
We assume this is true for the $t$-th round and prove this
true for the $(t+1)$-st round. 

By our induction hypothesis, the adversary could not have computed any blocks
for the $t$-th round using $d_t$. This means that the adversary
obtained $d_t$ at the beginning of the $(t+1)$-st round.
Then, suppose the adversary computes a block 
for the $(t + 1)$-st round. By~\cref{item:difficulty} of
Definition~\ref{def:vdfs}, the adversary could not have computed 
the output in time less than the
duration of a round,\footnote{Recall we set the
difficulty of the VDF to be the duration of a round.}
with all but negligible
probability in $\secp$. Thus, the adversary could not have mined
any blocks using $d_{t+1}$ until round $t+2$. Hence, no adversary
can use $d_{t+1}$ in mining any blocks during any round $t'' \leq
t+1$.
\end{proof}

\begin{proof}[Proof of~\cref{thm:property-vdf}]
Lemma~\ref{lem:future} states that the adversary cannot
mine a block with VDF output, $o_{t'}$,
from a round $t' > t$ where
$t$ is the current round. Thus, during round $t$, an
adversary can only use VDF outputs, $o_1, \dots, o_{t-1}$. Then, Lemma~\ref{lem:vdf-seq} and
Lemma~\ref{lem:adversary-two-blocks} ensure that no 
two blocks in a chain can contain the same VDF output. Then, the length of any chain accepted by an honest
node has length at most $t$ since it can contain at most
one block using each of the VDF outputs, $o_1, \dots, o_{t-1}$ (plus the genesis block).
\end{proof}
\fi

\end{document}